\documentclass[11pt,a4paper]{article}
\usepackage[left=2.5cm,right=2.5cm,bottom=3cm,top=3cm,bindingoffset=0.5cm]{geometry}

\usepackage[utf8]{inputenc}
\usepackage[american]{babel}
\usepackage{amsmath}
\usepackage[final,pdfusetitle, pdftex]{hyperref}
\hypersetup{pdfauthor={Birger Brietzke, S{\o}ren Fournais and Jan Philip Solovej}
}
\usepackage{amssymb}
\usepackage{amsthm}
\usepackage{fancyhdr}
\usepackage{graphicx}
\usepackage{curves}
\usepackage{subfigure}
\usepackage{boxedminipage}
\usepackage{cite}
\usepackage{lastpage}
\usepackage{textcase}
\usepackage{authblk}
\usepackage{color}
\usepackage{verbatim}
\usepackage{bbm}

\theoremstyle{definition} \newtheorem{definition}{Definition}[section]
\theoremstyle{plain} \newtheorem{theorem}[definition]{Theorem}
\theoremstyle{plain} \newtheorem{assumption}[definition]{Assumption}
\theoremstyle{plain} 
\theoremstyle{plain} \newtheorem{lemma}[definition]{Lemma}
\theoremstyle{plain} \newtheorem{corollary}[definition]{Corollary}
\theoremstyle{plain} \newtheorem{remark}[definition]{Remark}
\theoremstyle{definition} 
\theoremstyle{plain} 
\numberwithin{equation}{section}
\fancypagestyle{plain}{%
  \fancyhf{}%
}

\newcommand{\R}{\mathbb{R}}
\newcommand{\C}{\mathbb{C}}

\newcommand{\one}{\mathbbm{1}}

\newcommand{\cA}{\mathcal{A}}
\newcommand{\cB}{\mathcal{B}}
\newcommand{\bn}{{\mathord{b}}^{\phantom{*}}}
\newcommand{\rmu}{\rho_{\mu}}
\newcommand{\cT}{\mathcal{T}}

\newcommand{\abs}[1]{\lvert {#1} \rvert }

\DeclareMathOperator{\supp}{\mathrm{supp}}
\DeclareMathOperator{\Spec}{\mathrm{Spec}}

\renewcommand{\epsilon}{\varepsilon}
\renewcommand{\phi}{\varphi}

\newcommand{\constkinetickernel}{C_{\mathrm{kin}}}

\newcommand{\CA}{16\pi}

\usepackage{xifthen}% provides \isempty test
\newcommand{\place}[2]{%
	\ifthenelse{\isempty{#2}}%
	{({\bf p}: ${#1}$):\quad}% if #2 is empty
	{({\bf p}: ${#1}$, {\bf l}: ${#2}$):\quad}% if #2 is empty
}

\newcommand{\off}[1]{\left [ {#1}\right ]}

\newcommand{\vg}[1]{{v{#1}}}

\begin{document}

\title{A simple 2nd order lower bound to the energy of dilute Bose gases}
\date{}
\author[1]{Birger Brietzke}
\author[2]{S\o ren Fournais}
\author[3]{Jan Philip Solovej}
\affil[1]{\small{ Institute of Applied Mathematics\\Heidelberg University\\Im Neuenheimer Feld 205\\
69120 Heidelberg, Germany\\
brietzke@math.uni-heidelberg.de}}
\affil[2]{\small{Department of Mathematics, Aarhus University\\ Ny Munkegade 118\\ DK-8000 Aarhus C\\ Denmark\\
fournais@math.au.dk}}
\affil[3]{\small{Department of Mathematics\\ University of Copenhagen\\ Universitetsparken 5\\ DK-2100 Copenhagen \O\\Denmark\\
solovej@math.ku.dk}}
\maketitle

\begin{abstract}
For a dilute system of non-relativistic bosons interacting through a positive, radial potential $v$ with scattering length $a$
we prove that the ground state energy density satisfies the bound $e(\rho) \geq 4\pi a \rho^2 (1- C \sqrt{\rho a^3} \,)$.
\end{abstract}

\tableofcontents

\section{Introduction}
We study a system of $N$ interacting bosons in a large box $\Lambda \subset {\mathbb R}^3$ of volume $|\Lambda|$. For concreteness, we take 
$L>0$ and define
$\Lambda = [-L/2,L/2]^3$.
We are interested in the thermodynamic limit $N \rightarrow \infty$, $|\Lambda| \rightarrow \infty$ with density $\widetilde \rho := N/|\Lambda|$ fixed and small.

The Hamiltonian of the system is
\begin{align}\label{eq:Hamiltonian}
H_N := \sum_{i=1}^N -\Delta_i + \sum_{i<j} v(x_i-x_j),
\end{align}
on the symmetric (bosonic) space $\otimes_{s}^N L^2(\Lambda)$.
We take $H_N$ with Dirichlet boundary conditions to realize it as a self-adjoint operator.

We define the {\it ground state energy} of the system $E_0= E_0(N,\Lambda)$ to be
$$
E_0(N, \Lambda) := \inf \Spec H_N
$$
and the {\it thermodynamic ground state energy density} as
$$
e(\widetilde \rho\,) = \lim_{L\rightarrow \infty,\, N/|\Lambda| = \widetilde \rho} E_0(N, \Lambda) /L^3.
$$
Our main result, Theorem~\ref{thm:LHY-General}, is formulated in terms of the {\it scattering length} $a = a(v)$ of the potential $v$. The definition and useful properties of the scattering length will be given in Section~\ref{scattering} below.

\begin{theorem}\label{thm:LHY-General}
There exists a universal constant $C>0$ such that the following is true.

For all $v: {\mathbb R}^3 \rightarrow [0,+\infty]$ positive, radial, measurable, potentials with $v(x)=0$ for all $|x| > R$, we have
\begin{align}
\label{eq:LowerBound-General}
e(\widetilde{\rho}\,) \geq 4\pi \widetilde \rho\,^2 a \left(1 - C \big(\sqrt{\widetilde \rho a^3} + R^2a \widetilde{\rho}\,\big)\right).
\end{align}
\end{theorem}
Clearly $e(\widetilde{\rho})\geq 0$, so \eqref{eq:LowerBound-General} is trivially satisfied unless $R^2 a\widetilde{\rho}\leq C^{-1}$, which implies the usual diluteness criterion ${\widetilde{\rho} a^3\leq C^{-1}}$ by \eqref{eq: upper bound on scattering length finite range} below.
Potentials satisfying the conditions in Theorem~\ref{thm:LHY-General} in particular, by \eqref{eq: upper bound on scattering length finite range} below, have finite scattering length.

Theorem~\ref{thm:LHY-General} can be generalized towards potentials that do not have finite range, see Theorem~\ref{thm:LHYsizeVeryGeneral} below for details.

\begin{remark}
The result of Theorem~\ref{thm:LHY-General} can in particular be applied to the `hard core' potential, i.e.
\begin{align}
v(x) = \begin{cases}
0, & |x| \geq R, \\
+\infty, & |x| < R.
\end{cases}
\end{align}
In this case $a=R$, so the error term $R^2a \widetilde{\rho}$ is higher order.  For the hard core potential we get the lower bound
\begin{align}
\label{eq:LowerBound-hc}
e(\widetilde{\rho}\,) \geq 4\pi \widetilde \rho\,^2 a \left(1 - C \sqrt{\widetilde \rho a^3}\right).
\end{align}
Our result is the first rigorous lower bound on the hard core potential that gives the correct order for the correction term.
(See below for a further discussion of the expected correction term, the so-called Lee-Huang-Yang term \cite{LHY}).
\end{remark}

\begin{remark}[Definition of the Hamiltonian]
The operator $H_N$ is not immediately defined for potentials of the generality considered in Theorem~\ref{thm:LHY-General}, so here we give the details of the definition of $H_N$.

Consider the quadratic form 
\begin{align}
Q_N(\Psi) = \int_{\Lambda^N} \sum_{j=1}^N |\nabla_j \Psi|^2 + \sum_{i<j} v(x_i-x_j) |\Psi|^2\,dx,
\end{align}
defined on the domain
\begin{align}
D(Q_N):= \Big\{ \Psi \in H_0^1(\Lambda^N)\,\Big{|}\, \Big(\sum_{i<j} v(x_i-x_j)\Big)^{1/2} \Psi \in L^2(\Lambda^N) \Big\}.
\end{align}
Consider furthermore,
\begin{align}
{\mathcal H}_Q = \overline{ \{  \Psi \in L^2(\Lambda^N) \,|\, Q_N(\Psi) + \| \Psi \|_2^2 < \infty \}},
\end{align}
where the closure is taken in $L^2(\Lambda^N)$.
Being a closed subset of $L^2(\Lambda^N)$, ${\mathcal H}_Q$ is a Hilbert space in its own right having $D(Q_N)$ as a dense subspace.
Therefore, clearly $Q_N$ defines a densely defined, Hermitian, quadratic form in ${\mathcal H}_Q$. It is an exercise to check that $Q_N$ is a closed form and it therefore follows that $Q_N$ defines a unique self adjoint operator. The resulting operator is by definition $H_N$.

\end{remark}

The proof of Theorem~\ref{thm:LHY-General}, which is the first general lower bound containing the expected order of the correction term, will be given in Section~\ref{sec:Fock}. We will briefly review some previous results below.

The rigorous study of the ground state energy of the interacting boson problem has a long history\footnote{When comparing results in the literature, one should notice that some authors study the energy per particle, i.e. $\lim E_0(N, \Lambda) /N = {\widetilde \rho\,}^{-1} e(\widetilde{\rho}\,)$ instead of the energy density, leading, of course, to slightly different formulae.}.
Bogoliubov's theory \cite{Bog3} prescribes how to treat the weak coupling limit of interacting bosons. In the context of the present paper, this weak coupling limit corresponds to the dilute gas, i.e. the limit $\widetilde \rho \rightarrow 0$ under study.
The Bogoliubov theory actually gives much more detailed information, e.g. on the excitation spectrum, but the ground state energy is one of the simplest quantities on which to obtain rigorous information regarding the validity of Bogoliubov's approach.

It was Lenz \cite{Lenz} who proposed the leading order behavior in \eqref{eq:LowerBound}.
Dyson \cite{dyson} proved that the leading order of \eqref{eq:LowerBound} has the correct form. His upper bound provides the sharp constant, while his lower bound only captures the correct leading order and was completed by the corresponding lower bound 40 years later in \cite{LY}.
It is worth noticing that the upper bound by Dyson is valid for hard core potentials, whereas the improved lower bounds \cite{ESY,YY} to be discussed below all require additional regularity of the interaction.

To even higher precision, the energy density is expected to behave as
\begin{align}\label{eq:LHY}
e(\widetilde{\rho}\,) = 4\pi \widetilde \rho\,^2 a \bigg (1+  \frac{128}{15 \sqrt{\pi}} (\widetilde \rho a^3)^{1/2}
+ o(\widetilde{\rho} a^3)^{1/2}\bigg ).
\end{align}
The second term in \eqref{eq:LHY} is often referred to as the Lee-Huang-Yang term after \cite{LHY} but is also heuristically understandable from Bogoliubov's treatment. For this and other background information on the Bose gas we refer to \cite{LSSY}.

In \cite{ESY} an upper bound to $e(\widetilde{\rho}\,)$ is given which correctly reproduces the first term and the order of the second term in \eqref{eq:LHY}, however only giving the correct coefficient on the correction term in the additional limit of weak interaction.
From \cite{BogFuncII} (see also \cite{BogFuncI} for more information on the Bogoliubov functional) one can actually conclude that to get \eqref{eq:LHY} one needs to go beyond states that are quasi-free. Indeed, they prove that for the ground state energy the trial state in \cite{ESY} is essentially optimal among quasi-free states.

An upper bound consistent with \eqref{eq:LHY} has been proven in \cite{YY} using trial states that are not quasi-free.
As already mentioned, the improved upper bounds by \cite{ESY,YY} do not work in the hard-core case. For the hard core the best upper bound that we are aware of remains \cite{dyson} with an error term of size ${\mathcal O}((\widetilde{\rho} a^3)^{1/3})$ relative to the leading order term.

The asymptotic result \eqref{eq:LHY} has first been proven in cases where the interaction is scaled to become `soft' in a manner depending on $\widetilde \rho$ \cite{BS,GS}.
In the very recent paper \cite{FS} two of the authors prove the Lee-Huang-Yang result \eqref{eq:LHY} for general radial, $L^1$-potentials. However, this result is not uniform in $a^{-1} \int v$, in particular it does not apply to the hard-core potential.

The recent work \cite{BBCS,BBCS2} is also very relevant for 
\eqref{eq:LHY}, though they address the confined case in the Gross-Pitaevskii limit and not the thermodynamic limit.
Actually, the result obtained in \cite{BBCS} is after scaling very analogous to our analysis of the box Hamiltonian (see Theorem~\ref{thm:LHY-Box} below).
We have the additional difficulties that for our localized problem, we no longer have translation invariance nor a fixed number of particles.
Nevertheless, we believe that our method, at least for the ground state energy, is substantially shorter and simpler than the one of \cite{BBCS}, which also covers the excitation spectrum.

In the papers \cite{LS,LS2} the Bogoliubov approximation is proved to give the right result in the setting of the ground state energy of a charged gas. In the present paper we use the general strategy laid out in those papers.

\paragraph{Notation.}
We use the convention that integrals are over all of ${\mathbb R}^3$ unless the domain of integration is explicitly specified.

\paragraph{Organization of the paper.}
The paper is organized as follows. 
In Section~\ref{simplified} we reduce the problem to the study of the case where the potential $v$ satisfies an $L^1$ condition.
The remainder of the paper is carried out under this assumption.
We start by recalling basic definitions and results about the scattering length $a$ and related quantities in Section~\ref{scattering}. Then, in Section~\ref{sec:Fock} we reformulate the many-body problem in a Fock space setting---see in particular Theorem~\ref{thm:LHY-Background}. This will allow us to use a simple version of Bogoliubov's theory in which the number of particles is not fixed.
In Section~\ref{sec:Localization} a localization to boxes of length scale $\ell \approx 1/\sqrt{\rmu a}$ is carried out. This is an important and delicate step since our proof requires the localization to be carried out in such a way as not to lose the Neumann gap.
The final result of the section is Theorem~\ref{thm:CompareBoxEnergy} whereby all we have to study is the ground state energy of one fixed box, which is the purpose of the remainder of the article.
The main work is carried out in Section~\ref{sec:BoxEnergy}.
In Lemma~\ref{lem:SplittingInQs} we estimate the terms in the Hamiltonian that are not quadratic in excitations out of the constant function. The important point here is inspired by the analysis of the Bogoliubov functional and consists of `completing a square' relative to the quartic term in the excitations.
The terms remaining are quadratic, thus allowing us to second quantize and use the Bogoliubov method. This we carry out in Section~\ref{sec:Bogoliubov}. Finally, in Section~\ref{sec:TotalEnergy} we put the pieces together to prove Theorem~\ref{thm:LHY-Box} which, using the first sections, implies Theorem~\ref{thm:LHY-General}.

\paragraph{Acknowledgements.}
BB and JPS were partially supported by the Villum Centre of Excellence for
the Mathematics of Quantum Theory (QMATH) and the ERC Advanced grant 321029.
BB also gratefully acknowledges support from the DFG, Grant number AOBJ 643360 KN 102013-1.
SF was partially supported by a Sapere Aude grant
from the Independent Research Fund Denmark, Grant number
DFF--4181-00221.

\section{The simplified result}\label{simplified}
The main work of the manuscript will be carried out under an $L^1$ assumption on the potential $v$. This will allow us to have a well-defined Fourier transform of $v$ (more precisely of the function $g$ related to the scattering length and defined in \eqref{eq:Defg} below). However, bounds will be uniform in the $L^1$-norm of $v$. In this section we will state the simplified result---Theorem~\ref{thm:LHY} below---and show how the main result Theorem~\ref{thm:LHY-General} follows from it.

\begin{assumption}\label{assump:v}
The potential $v$ is non-negative and spherically symmetric, i.e. $v(x) = v(|x|)\geq 0$, and integrable with compact support. We fix $R>0$ such that $\supp v \subset B(0,R)$.
\end{assumption}

\begin{theorem}\label{thm:LHY}
There exists a universal constant $C>0$ such that if the potential $v$ satisfies Assumption~\ref{assump:v}, then
\begin{align}
\label{eq:LowerBound}
e(\widetilde{\rho}\,) \geq 4\pi \widetilde \rho\,^2 a \left(1 - C \big(\sqrt{\widetilde \rho a^3} + R^2a \widetilde{\rho}\,\big)\right).
\end{align}
\end{theorem}

We now prove that Theorem~\ref{thm:LHY-General} follows from Theorem~\ref{thm:LHY}.
\begin{proof}[Proof of Theorem~\ref{thm:LHY-General}]
Define the sequence of potentials
\begin{align}
v_n(x) := \min\{ v(x), n\}.
\end{align}
Furthermore, let $a_n$ be the scattering length of $v_n$
and let $e_n(\widetilde{\rho}\,)$ be the ground state energy density of the Hamiltonian \eqref{eq:Hamiltonian} with potential $v_n$.
From the definition of the ground state energy density it is clear that $e(\widetilde{\rho}\,) \geq e_n(\widetilde{\rho}\,)$. Furthermore $a_n$ (is well-defined and) satisfies $a_n \leq a$ by the monotonicity of the scattering length.

Therefore, \eqref{eq:LowerBound-General} follows from \eqref{eq:LowerBound} (using that the constant $C$ is independent of $n$) and the fact that $a_n \rightarrow a$, which follows from Lemma~\ref{lem:ScattAppr} below.
\end{proof}

In Definition~\ref{def: scattering length for infinite range} below we define the scattering length for potentials with infinite range.
This is used in the following theorem where we give a lower bound for potentials with infinite range. If the potential is sufficiently soft at infinity, then our bound on the ground state energy density is compatible with \eqref{eq:LHY}, i.e. the Lee-Huang-Yang formula. Our proof is based on a limiting argument and therefore we introduce the following notation for a given potential $v$ and $R,R'\geq 0$:
\begin{align}\label{def: v leq R and v>R}
v_{\leq R}(x):=v(x) \one_{\{|x|\leq R\}},\qquad 
v_{> R}(x):= v(x) \one_{\{|x| > R\}}, \qquad
v_{>R}^{R'}(x):=v(x)\one_{\{R<|x| \leq R'\}}.
\end{align}
\begin{theorem}[General version without finite range]\label{thm:LHYsizeVeryGeneral}
There exists a universal constant $C>0$ such that the following is true.

Suppose that $v\in L^1_{\mathrm{loc}}({\mathbb R}^3)$ is positive, radial, and has a finite scattering length $a$.
Then, for all $R\geq 0$ and with $v_{\leq R}$ and $v_{> R}$ defined in \eqref{def: v leq R and v>R}, the scattering lengths $a(v_{\leq R})$ and  $a(v_{> R})$ are finite and
\begin{align}
\label{eq:LowerBound-Very-General}
e(\widetilde{\rho}\,) \geq 4\pi \widetilde \rho\,^2  \left(a - a(v_{> R}) - Ca \big(\sqrt{\widetilde \rho a^3} + R^2a \widetilde{\rho}\,\big)\right).
\end{align}

In particular, if $v$ satisfies that $a(v_{> R_{\widetilde{\rho}}}) \leq C a\sqrt{\widetilde \rho a^3}$, where $R_{\widetilde{\rho}} :=  (\widetilde{\rho} a^3)^{-\frac{1}{4}} a$, then
\begin{align}
\label{eq:LowerBound-Very-General}
e(\widetilde{\rho}\,) \geq 4\pi \widetilde \rho\,^2 a \left(1 - 3C \sqrt{\widetilde \rho a^3} \right).
\end{align}
\end{theorem}

\begin{remark}
One can think of the last hypothesis of Theorem~\ref{thm:LHYsizeVeryGeneral} as an assumption on the decay of $v$. Recall the general inequality \eqref{eq:ScatInt} below for the scattering length. Applying this to $v_{> R_{\widetilde{\rho}}}^{n}$, where $R_{\widetilde{\rho}} :=  (\widetilde{\rho} a^3)^{-\frac{1}{4}} a$, and taking the limit in $n$ we see that $a(v_{> R_{\widetilde{\rho}}}) \leq C a\sqrt{\widetilde \rho a^3}$,  if $v$ decays like $|x|^{-5}$ (outside a compact set).
So \eqref{eq:LowerBound-Very-General} is valid for potentials with this type of decay.
\end{remark}

\begin{proof}[Proof of Theorem~\ref{thm:LHYsizeVeryGeneral}]
The finiteness of the scattering lengths $a(v_{\leq R}), a(v_{> R})$ follows from \eqref{eq:ScatteringAdditive} below.
Denote by $e_R(\widetilde{\rho})$ the thermodynamic ground state energy density in the potential $v_{\leq R}$. Then, using Theorem~\ref{thm:LHY-General} for $e_R(\widetilde{\rho}\,)$ and the monotonicity of the energy, 
\begin{align}
e(\widetilde{\rho}\,) \geq e_R(\widetilde{\rho}\,) \geq 4 \pi a(v_{\leq R})  \widetilde{\rho}\,^2 \left( 1 - C\Big(\sqrt{\widetilde{\rho} a(v_{\leq R})^3} + R^2 a(v_{\leq R}) \widetilde{\rho}\Big)\right).
\end{align}
The estimate \eqref{eq:LowerBound-Very-General} now follows using 
the monotonicity of the scattering length and that $a\geq a(v_{\leq R})  \geq a - a(v_{> R})$ by \eqref{eq:ScatteringAdditive} below.
\end{proof}

\section{The scattering length}\label{scattering}
In this short section we establish notation and results concerning the scattering length and associated quantities. 
We refer to \cite[Appendix C]{LSSY} for more details.

\begin{definition}\label{def: scattering length finite range}
For a potential $v: {\mathbb R}^3 \rightarrow [0,+\infty]$ positive, radial, measurable such that $v(x)=0$ for all $|x| > R$, the {\it scattering length} $a=a(v)$ is defined by
\begin{align}\label{eq:ScatLengthMin}
\frac{4\pi a}{1- a/\widetilde{R}} = \inf\{ \int_{\{|x|\leq \widetilde{R}\} }|\nabla \phi(x)|^2 + \frac{1}{2}v(x) |\phi(x)|^2\,dx
\}.
\end{align}
Here $\widetilde{R} > R$ is arbitrary and the infimum is taken over
\begin{align}
\{ \phi \in H^1(B(0,\widetilde{R}): \phi_{|x|=\widetilde{R}} =1
\}.
\end{align}
\end{definition}

An analysis of the minimisation problem shows that $a$ is independent of the choice of $\widetilde{R} > R$ and satisfies \begin{align}
\label{eq: upper bound on scattering length finite range}
a\leq R.\end{align}
Also, one immediately sees from the minimization problem that
\begin{align}\label{eq:ScatInt}
a \leq \frac{1}{8\pi} \int v.
\end{align}

Furthermore, there is a unique minimizer $\phi_{v,\widetilde{R}}$ in \eqref{eq:ScatLengthMin}, which is radial and satisfies
$\phi_{v,\widetilde{R}}= (1-a/\widetilde{R})^{-1} \phi_v(x)$.
Here the function $\phi_v(x)$ is independent of $\widetilde{R}$, radial, non-negative, monotone non-decreasing as a function of $|x|$ and satisfies (in the sense of distributions on the set where $v$ is $L^1_{\mathrm{loc}}$)
\begin{align}\label{eq: R tilde scattering equation}
-\Delta \phi_v + \frac{1}{2} v \phi_v = 0,\qquad \text{and, for }|x| \geq R,  \qquad \phi_v(x) = 1 - \frac{a}{|x|}.
\end{align}
Furthermore, if $\widetilde{v}$ satisfies the requirements in Definition~\ref{def: scattering length finite range} and $ v\geq \widetilde{v}\geq 0$, then (see \cite[Lemma C.2.C]{LSSY})
\begin{align}\label{eq: relation for v<v'}
\phi_v\leq \phi_{\widetilde{v}}\qquad \textrm{and}\qquad 
a(v)\geq a(\widetilde{v})\geq 0.
\end{align}
We will need the following approximation result in our proof of Theorem~\ref{thm:LHY-General}.

\begin{lemma}\label{lem:ScattAppr}
Suppose $v: {\mathbb R}^3 \rightarrow [0,+\infty]$ positive, radial, measurable such that $v(x)=0$ for all $|x| > R$ and define
\begin{align}
v_n(x) := \min\{ v(x), n\}.
\end{align}
Then the sequence of scattering lengths $\{a(v_n)\}_n$ satisfies $a(v_n) \nearrow a(v)$ as $n\rightarrow \infty$.
\end{lemma}

\begin{proof}
It is immediate from the minimization problem in \eqref{eq:ScatLengthMin} that $a(v_n)$ is a monotone non-decreasing sequence and that $a(v_n) \leq a(v)$ for all $n$.

To prove the convergence of $a(v_n)$ towards $a(v)$ we consider the minimizer $\phi_{v,\widetilde{R}}$ of \eqref{eq:ScatLengthMin} for the potential $v$.
Clearly, $\phi_{v,\widetilde{R}}=0$ a.e. on $\{x:v(x) = \infty\}$, so
\begin{align}\label{eq:singularset}
\int_{\{x:v(x) = \infty\}} v(x) |\phi_{v,\widetilde{R}}|^2\,dx = 0.
\end{align}
Therefore, using $\phi_{v,\widetilde{R}}$ as a trial state for the minimization problem for $a(v_n)$, we find
\begin{align}
\frac{4\pi a(v)}{1- a(v)/\widetilde{R}} - \frac{4\pi a(v_n)}{1- a(v_n)/\widetilde{R}} 
\leq \frac{1}{2} \int_{\{x:v(x) \geq n\}} (v(x)-n) |\phi_{v,\widetilde{R}}|^2\,dx \rightarrow 0,
\end{align}
in the limit $n\rightarrow \infty$, where we used \eqref{eq:singularset} to pass to the limit.
\end{proof}%

\begin{definition}\label{def: scattering length for infinite range}
Let $v: {\mathbb R}^3 \rightarrow [0,+\infty]$ be positive, radial, measurable.
Then, the scattering length $a(v_{\leq R})$, with $v_{\leq R}(x)$ defined in \eqref{def: v leq R and v>R}, for all $R\geq 0$ is well-defined, in the sense of Definition~\ref{def: scattering length finite range}.
If the sequence $\{a(v_{\leq R})\}_R$ converges, then we say that the potential $v$ has {\it scattering length} 
\begin{align}
a(v):=\underset{R\to \infty}{\lim}a(v_{\leq R}).
\end{align}

\end{definition}
Note that this definition agrees with Definition~\ref{def: scattering length finite range} if $v$ has compact support. For $0\leq v\in L^1_{\mathrm{loc}}({\mathbb R}^3)$ and $0\leq R_1\leq R_2$ we have the pointwise inequality $v_{\leq R_1}\leq v_{\leq R_2}$. It is therefore immediate from \eqref{eq: relation for v<v'} that $a(v_{\leq R})$ is a non-decreasing function of $R\geq 0$, implying that $\lim_{R \rightarrow \infty} a(v_{\leq R})$ exists as soon as $\{a(v_{\leq R})\}_R$ is bounded from above.
\begin{lemma}\label{lem:ScatteringAdditive}
Suppose that $v\in L^1_{\mathrm{loc}}({\mathbb R}^3)$ is positive, radial, and has a finite scattering length $a(v)$ and define $v_{\leq R}$ and $v_{> R}$ via \eqref{def: v leq R and v>R}.
Then, for all $R\geq 0$,
\begin{align}\label{eq:ScatteringAdditive}
0\leq\max\{ a(v_{\leq R}), a(v_{> R})\} \leq a(v) \leq a(v_{\leq R}) + a(v_{> R}).
\end{align}
\end{lemma}
\begin{proof}
We fix $R\geq 0$. Recall the definition of $v_{>R}^{n}$ in \eqref{def: v leq R and v>R}. Since $0\leq v_{>R}^{n} \leq v_{\leq n}$ it follows easily from the existence of $a(v)$ that the scattering length of $v_{> R}$ exists and satisfies $0\leq a(v_{>R})\leq a(v)$. Since, for all $n\geq R$, we have $0\leq v_{\leq R}\leq v_{\leq n}$, taking a limit gives $0\leq a(v_{\leq R})\leq a(v)$.
To obtain the last inequality in \eqref{eq:ScatteringAdditive} we use that $a(v)$ in Definition~\ref{def: scattering length for infinite range} is defined via finite range potentials. With $\varphi_n:=\varphi_{v_{\leq n}}$ denoting the function defined around \eqref{eq: R tilde scattering equation} (and satisfying that equation since $v\in L^1_{\mathrm{loc}}({\mathbb R}^3)$) we get
\begin{align}\label{eq: scattering limit}
8 \pi a = 8\pi \lim_{n\rightarrow \infty} a(v_{\leq n}) =  \lim_{n\rightarrow \infty}  \int v_{\leq n} \varphi_n.
\end{align}
For $R\geq 0$ still fixed and $n\geq R$ we split the integral in \eqref{eq: scattering limit} as
\begin{align}
 \int v_{\leq n} \varphi_n = \int v_{\leq n} \one_{\{|x| \leq R\}} \varphi_n +  \int v_{\leq n} \one_{\{|x| > R\}} \varphi_n=  \int v_{ \leq R} \varphi_n +  \int v_{>R}^{n} \varphi_n.
\end{align}
For all $n \geq R$ we have $v_{ \leq R} \leq v_n$ and therefore, using \eqref{eq: relation for v<v'},
\begin{align}
\int v_{ \leq R} \varphi_n \leq \int v_{ \leq R}  \varphi_{v_{ \leq R}} = 8 \pi a( v_{ \leq R}).
\end{align}
Since $ v_{>R}^{n}\leq v_{\leq n}$ one more application of \eqref{eq: relation for v<v'} gives $\varphi_{v_{>R}^n}\geq \varphi_n$ and therefore
\begin{align}
\int v_{>R}^{n} \varphi_n \leq \int v_{>R}^n \varphi_{v_{>R}^n} = 8\pi a(v_{>R}^n).
\end{align}
By definition we have $\underset{n\to \infty}{\lim}a( v_{>R}^n)= a(v_{>R})$ and thus we get the last inequality in \eqref{eq:ScatteringAdditive} by combining the inequalities above and going to the limit.
\end{proof}
\subsection{Scattering quantities for $L^1$-potentials}
We proceed to introduce some notation for quantities related to the scattering length which will be used in the remainder of the paper.
For potentials satisfying Assumption~\ref{assump:v} we reformulate the scattering equation in \eqref{eq: R tilde scattering equation} as
\begin{align}\label{eq:Scattering2}
(-\Delta + \frac{1}{2} v(x) )(1-\omega(x)) =0,\qquad \text{ with } \omega \rightarrow 0, \text{ as } |x| \rightarrow \infty.
\end{align}
The solution $\omega$ to this equation satisfies that $\omega(x) = a/|x|$ for $x$ outside $\supp\, v$. We will refer to $\omega$ as the {\it scattering solution}. Furthermore, $\omega$ is radially symmetric and non-increasing with
\begin{align}
0\leq \omega(x)\leq 1.\label{omegabounds}
\end{align}
We introduce the function
\begin{align}\label{eq:Defg}
g= v(1-\omega).
\end{align}
The scattering equation can be reformulated as
\begin{align}
\label{eq:Scattering3}
-\Delta \omega = \frac{1}{2} g.
\end{align}
From this we deduce, using the divergence theorem, that
\begin{align}
a = (8\pi)^{-1} \int g,
\end{align}
and that the Fourier transform satisfies
\begin{align}\label{es:scatteringFourier}
\widehat{\omega}(k) = \frac{\hat{g}(k)}{2 k^2}.
\end{align}

\section{An equivalent problem on Fock space}\label{sec:Fock}
For convenience we reformulate the problem on Fock space.

Consider, for given $\rmu >0$, the following operator ${\mathcal H}_{\rmu}$ on the symmetric Fock space ${\mathcal F}_{\rm s}(L^2(\Lambda))$. The operator ${\mathcal H}_{\rmu}$ commutes with particle number and satisfies, with ${\mathcal H}_{\rmu,N}$ denoting the restriction of ${\mathcal H}_{\rmu}$ to the $N$-particle subspace of ${\mathcal F}_{\rm s}(L^2(\Lambda))$,
\begin{align}\label{eq:BackgroundH}
{\mathcal H}_{\rmu,N} & =
\sum_{i=1}^N \left( -\Delta_i - \rmu \int_{{\mathbb R}^3} g{(x_i-y)}\,dy \right)
+ \sum_{i<j} v(x_i-x_j)
 \nonumber \\
&= \sum_{i=1}^N -\Delta_i 
+ \sum_{i<j} v{(x_i-x_j)}
- 8\pi a \rmu N.
\end{align}
Notice that the new term in ${\mathcal H}_{\rmu,N}$ plays the role of a chemical potential justifying the notation.

Define the corresponding ground state energy density,
\begin{align}
e_0(\rmu):= \lim_{|\Lambda| \rightarrow \infty} |\Lambda|^{-1} \inf_{\Psi \in {\mathcal F_{\rm s}}\setminus \{0\}} \frac{\langle \Psi,  {\mathcal H}_{\rmu} \Psi \rangle}{\| \Psi \|^2}.
\end{align}

\noindent We formulate the following result, which will be a consequence of Theorems~\ref{thm:CompareBoxEnergy} and \ref{thm:LHY-Box}.

\begin{theorem}\label{thm:LHY-Background}
There exists a universal constant $C>0$ such that the following is true.
Suppose $v$ satisfies Assumption~\ref{assump:v} and that $\rmu a R^2 \leq C^{-1}$, $\rmu a^3 \leq C^{-1}$.
Then the ground state energy density of ${\mathcal H}_{\rmu}$ satisfies that
\begin{align}
e_0(\rmu) \geq 4\pi \rmu^2 a \left(- 1 - C (\rmu a^3) ^{1/2} - C \rmu a R^2
\right).
\end{align}
\end{theorem}

\begin{proof}[Proof of Theorem~\ref{thm:LHY}]
It is easy to deduce Theorem~\ref{thm:LHY} from Theorem~\ref{thm:LHY-Background}.
Clearly $e(\widetilde \rho\,) \geq 0$, so it suffices to consider the case where $\widetilde \rho a R^2 \leq C^{-1}$.

By inserting the ground state of $H_N$ as a trial state in ${\mathcal H}_{\rmu}$ one gets in the thermodynamic limit that
\begin{align}\label{eq:CompareGC}
e(\widetilde \rho\,) \geq e_0(\rmu) +  \widetilde \rho \rmu\int g=e_0(\rmu) + 8\pi a \widetilde \rho \rmu .
\end{align}
For all $\rmu\in (0,\widetilde \rho\,]$ we may, in view of \eqref{eq: upper bound on scattering length finite range}, insert the lower bound from Theorem~\ref{thm:LHY-Background} into \eqref{eq:CompareGC}, which yields
\begin{align}\label{eq:Elem}
e(\widetilde \rho\,) &\geq 4\pi a \left[ -  \rmu^2 - C \rmu^2 (\rmu a^3)^{1/2}  - C \rmu^2 a R^2+ 2 \widetilde \rho \rmu
\right].
\end{align}
At this point we can choose $\rmu =  \widetilde \rho$ to get \eqref{eq:LowerBound}.
\end{proof}
\section{Reduction to a small box}\label{sec:Localization}
\subsection{Setup and notation}
The main part of the analysis will be carried out on a small box of size 
\begin{align}\label{eq:def_ell}
\ell := K(\rmu a)^{-1/2},
\end{align}
for some $K>0$ to be chosen sufficiently {\it small} but independent of $\rho_{\mu}$.
In this section we will carry out that localization. The main result is given at the end of the section as Theorem~\ref{thm:CompareBoxEnergy} which states that for a lower bound it suffices to consider a `box energy', i.e. the ground state energy of a Hamiltonian localized to a box of size $\ell$.
For convenience, in Theorem~\ref{thm:LHY-Box} we state the bound on the box energy that will suffice in order to prove Theorem~\ref{thm:LHY-Background}.

\medskip

Let $\chi \in C^{\infty}_c({\mathbb R}^3)$ be an even localization function, satisfying
\begin{align}
0 \leq \chi, \qquad \int \chi^2 = 1, \qquad \supp \chi \subset [-1/2, 1/2]^3.\label{def: chi}
\end{align}
The function $\chi$ will be fixed all through the paper. We will not try to optimize constants in the choice of $\chi$.

We define
\begin{align}
\chi_B(x):=\chi(\frac{x}{\ell})
\end{align}
and, for given $u \in {\mathbb R}^3$, 
\begin{align}
\chi_u(x) := \chi(\frac{x}{\ell}-u).
\end{align}
Notice that $\chi_u$ localizes to the box $B(u) := \ell u + [-\ell/2,\ell/2]^3$. 

We will also need the sharp localization function $\theta_u$ to the box $B(u)$, i.e.
\begin{align}
\theta_u := \one_{B(u)}.
\end{align}

Define $P_u, Q_u$ to be the orthogonal projections in $L^2({\mathbb R}^3)$ defined by
\begin{align}\label{def: projections}
P_u \varphi := \ell^{-3} \langle \theta_u, \varphi\rangle \theta_u, \qquad  Q_u \varphi:= \theta_u \varphi - \ell^{-3} \langle \theta_u, \varphi \rangle \theta_u.
\end{align}
Define furthermore
\begin{align}\label{eq:3.5}
W(x) := \frac{\vg{(x)}}{\chi*\chi(x/\ell)}.
\end{align}
Since $\chi * \chi(0) = 1$ by \eqref{def: chi}, we have
\begin{align}
\left| \chi*\chi(x) - 1\right| \leq \frac{1}{2}, \qquad \text{ for all } |x| \leq D,
\end{align}
where $D$ only depends on $\chi$.
Therefore, $W$ is well-defined by the finite range of $v$, for $R/\ell \leq D$,
i.e., for 
\begin{align}\label{eq:R/ell}
\rmu a R^2 \leq (KD)^2.
\end{align}
Note that this is no real restriction, as we mentioned after the statement of Theorem~\ref{thm:LHY-General}.
Define the localized potentials
\begin{align}
w_u(x,y) := \chi_u(x) W(x-y) \chi_u(y), \qquad w(x,y) := w_{u=0}(x,y).
\end{align}
Notice the translation invariance,
\begin{align}\label{eq:transInv}
w_{u+\tau}(x,y) = w_u(x-\ell \tau,y-\ell \tau).
\end{align}
For some estimates it is convenient to invoke the scattering solution and thus we introduce the notation, which again is well-defined for $\rho_\mu$ sufficiently small,
\begin{align}\label{def: W_2}
W_1(x) :={}& W(x) (1-\omega(x)) = \frac{g(x)}{\chi*\chi(x/\ell)},\qquad w_1(x,y):=w(x,y)(1-\omega(x-y))\\
W_2(x) :={}& W(x) (1-\omega^2(x)) = \frac{g(x)+g\omega(x))}{\chi*\chi(x/\ell)},\qquad w_2(x,y):=w(x,y)(1-\omega^2(x-y)).
\end{align}

For $\rmu$ sufficiently small a simple change of variables yields, for all $u\in{\mathbb R}^3$, the identities
\begin{align}\label{eq:DefU}
\frac{1}{2} \ell^{-6} \iint_{{\mathbb R}^3\times{\mathbb R}^3} \chi(\frac{x}{\ell})\chi(\frac{y}{\ell}) W_1(x-y)\, dx\,dy
&= \frac{1}{2} \ell^{-6} \iint_{{\mathbb R}^3\times{\mathbb R}^3} w_{1,u}(x,y)\,dx\,dy \nonumber \\
&=\frac{1}{2} \ell^{-3} \int g = 4 \pi \frac{a}{\ell^3} = \frac{1}{2} \ell^{-3} \widehat{g}(0)
\end{align}
and
\begin{align}\label{eq:w2int}
\frac{1}{2} \ell^{-6} \iint_{{\mathbb R}^3\times{\mathbb R}^3} w_2(x,y)\,dx\,dy 
&=\frac{1}{2}\ell^{-3} (\widehat{g}(0) + \widehat{g \omega}(0)).
\end{align}

\begin{lemma}\label{lem:IntEst}
There exists a constant $C$ (depending on $\chi$) such that for all $u$, 
and for all $\rmu$ such that \eqref{eq:R/ell} is satisfied, we have
\begin{align}
\max_x \int w_{1,u}(x,y)\,dy \leq C a,\label{eq:Pointwise_w1}\\
\max_x \int w_{2,u}(x,y)\,dy \leq C a,\label{eq:Pointwise_w2}
\end{align}
and furthermore, for all $x$,
\begin{align}\label{eq:W1-g}
g(x) \leq W_1(x) \leq g(x)(1 + C(R/\ell)^2).
\end{align}
\end{lemma}

\begin{proof}
By translation invariance, it suffices to consider $u=0$.
By definition,
\begin{align}
\int w_{2}(x,y)\,dy 
\leq 2\int w_{1}(x,y)\,dy= 2\chi(x/\ell) \int \frac{g{(x-y)}}{\chi*\chi((x-y)/\ell)} \chi(y/\ell)\,dy.
\end{align}
Since $\supp v \subset B(0,R)$, $\chi*\chi(0) = 1$ and $\chi*\chi$ is even, we get
\begin{align*}
0 &\leq \int w_{1}(x,y)\,dy=\int_{\abs{x-y}<R} w_{1}(x,y)\,dy\\
& \leq (1 + C (R/\ell)^2)^{-1} \chi(x/\ell) \int g{(x-y)} \chi(y/\ell)\,dy 
\leq {\|\chi}\|_\infty^2 (1 + C (R/\ell)^2)^{-1} 8 \pi a.
\end{align*}
The proof of \eqref{eq:W1-g} is similar and will be omitted.
\end{proof}

\subsection{Localization of the potential energy}
\begin{lemma}[Localization of potential energy]\label{lem:LocPotential}
If \eqref{eq:R/ell} is satisfied, we have for all $x_1,\ldots,x_N \in \Lambda$
\begin{align}
\sum_{i=1}^N - \rmu &\int g{(x_i-y)}\,dy + \sum_{i<j} \vg{(x_i-x_j)}  \nonumber \\
&=
\int_{\ell^{-1}(\Lambda+B(0,\ell/2))} \Big[
\sum_{i=1}^N - \rmu \int w_{1,u}(x_i,y)\,dy + \sum_{i<j} w_u(x_i,x_j) 
\Big]\,du .
\end{align}
\end{lemma}
\begin{proof}
We calculate, using $x_i, x_j \in \Lambda$,
\begin{align}\label{eq:ObtainingMeanField}
\int_{\ell^{-1}(\Lambda+B(0,\ell/2))} w_u(x_i,x_j)  \,du 
&=\int_{\ell^{-1}(\Lambda+B(0,\ell/2))} \chi(\frac{x_i}{\ell}-u)\chi(\frac{x_j}{\ell}-u)  \,du \,W(x_i-x_j)\nonumber \\
&= \vg{(x_i-x_j)}.
\end{align}
Here we used that if $\|x_i-x_j\|\leq R$ and $\rho_\mu$ is sufficiently small so that  \eqref{eq:R/ell} is satisfied, then the $u$-integral gives the (non-zero) convolution, which is the denominator in $W$.
The other term is similar.
\end{proof}

\subsection{Localization of the kinetic energy}\label{sec:kinloc}
In this subsection we prove a localization estimate on the kinetic energy in the box $B(u)$
centered at $\ell u$. The localized kinetic energy operator stems from Lemma~\ref{lem:LocKinEn-2} below and becomes
\begin{align}\label{eq:DefTu}
\mathcal{T}_u  := Q_u\left[
\chi_u 
\left( -\Delta- \constkinetickernel \ell^{-2} \right)_{+}
\chi_u
+
b \ell^{-2}   \right] Q_u,
\end{align}
where $b,\constkinetickernel>0$ are universal constants.\\
\indent Note that $\cT_u$ vanishes on constant functions. The last term in
$\cT_u$ will control the gap in the kinetic energy, i.e. on functions
orthogonal to constants in the box, $\cT_u$ is bounded below by at
least $b\ell^{-2}$.
A key result to obtain \eqref{eq:DefTu} is the lemma below.
\begin{lemma}[Abstract kinetic energy localization]\label{lm:abskinloc}
Let $\mathcal{K}:\R^3\to[0,\infty)$ be a symmetric, polynomially bounded, continuous function, and define the operator $T$ on $L^2(\R^3)$ by
\begin{equation}\label{eq:abskinloc}
T=\int_{\R^3} Q_u\chi_u \mathcal{K}(-i\ell\nabla)\chi_uQ_u  \,du,
\end{equation}
where $\chi_u$ is considered as a multiplication operator in configuration space.
This $T$ is translation invariant, i.e. a multiplication operator in Fourier space $T=F(-i\ell\nabla)$, with 
\begin{equation}\label{eq:absF}
F(p)=(2\pi)^{-3}\mathcal{K}*|\widehat\chi|^2(p)
-2(2\pi)^{-3}\widehat\theta(p)\widehat\chi*(\mathcal{K}\widehat\chi)(p)+(2\pi)^{-3}\left(\int
\mathcal{K}|\widehat\chi|^2\right)\widehat\theta(p)^2. 
\end{equation}
In particular, we have $F(0)=0$, $F\geq 0$ and $\nabla F(0)=0$. 
\end{lemma}

\begin{remark}\label{rm:LimitReg}
	For simplicity, we have chosen to assume that $\chi \in C^{\infty}_c$ whereby $\widehat{\chi}$ has fast decay.
	The same method works for localization functions with less regularity, the important assumption for Lemma~\ref{lm:abskinloc} being that the integral $\int
	\mathcal{K}|\widehat\chi|^2$ converges. In the accompanying paper \cite{BS} it will be important to use this flexibility.
\end{remark}

\begin{proof}
By a simple scaling it is enough to consider $\ell=1$.
	This is a straightforward calculation. Note that $Q_u$ has the integral kernel $\theta_u(y)\off{\delta(y-x)- \one}\theta_u(x)$.
	If we denote by $\check{\mathcal{K}}$ the inverse Fourier transform of $\mathcal{K}$ in the sense of a tempered distribution, then the integral kernel of the operator
	$
	Q_u\chi_u \mathcal{K}(-i\nabla)\chi_uQ_u
	$
	is given by 	
\begin{align*}
		&\chi_u(x)\check{K}(x-y)\chi_u(y)-\chi_u(x)[\check{\mathcal{K}}*\chi_u](x)\theta_u(y) \\
		&\qquad 
		-\theta_u(x)[\check{\mathcal{K}}*\chi_u](y)\chi_u(y)+
		\theta_u(x)\langle \chi_{u}|\mathcal{K}(-i\nabla)\chi_{u}\rangle\theta_u(y).
\end{align*}
Thus the integral kernel of $\int Q_u\chi_u \mathcal{K}(-i\nabla)\chi_uQ_u \,d u $ is
given by 
\begin{align*}
([\chi*\chi]\check{\mathcal{K}})(x-y)-2\left(\chi[\check{\mathcal{K}}*\chi]\right)*\theta(x-y)
+(2\pi)^{-3}\left(\int \mathcal{K}(p)\widehat\chi(p)^2\,d p\right)\theta*\theta(x-y),
\end{align*}
where we used that $\int \mathcal{K}(p)\widehat\chi(p)^2\,d p$ is finite by the choice of $\mathcal{K}$ and the decay of $\widehat\chi$.
We arrive at the expression for $F$ by calculating the inverse Fourier
transform.
The fact that $F(0)=0$ follows since $\widehat\theta(0)=\int\theta=1$
and
\begin{align*}
(2\pi)^3F(0)=2\left(\int \mathcal{K}\widehat \chi^2\right)(1-\widehat\theta(0))^2=0.
\end{align*}
That $F\geq 0$ is a direct consequence of \eqref{eq:abskinloc} since $\mathcal{K}$ is positive. Because $F$ is differentiable it follows that $\nabla F(0)=0$.
\end{proof}

With $\ell=1$ this lemma is similar to the generalized IMS localization formula
$$
\int_{\R^3} \chi_u \mathcal{K}(-i\nabla)\chi_u\,d u=(2\pi)^{-3}\mathcal{K}*|\widehat\chi|^2,
$$
where $\mathcal{K}(p)=p^2$ gives the standard IMS formula since then
$(2\pi)^{-3}\mathcal{K}*|\widehat\chi|^2=p^2+\int|\nabla\chi|^2$.

\begin{corollary}\label{eq:Quav} With the same notation as above we have
	that
	\begin{equation}\label{eq:averagedneumangap}
	\int_{\R^3}Q_u \,d u=1-\widehat\theta(-i\ell\nabla)^2,
	\end{equation}
	i.e. the operator $\int_{\R^3}Q_u \,d u$ is the multiplication operator
	in Fourier space given by $1-\widehat\theta(\ell p)^2$.
\end{corollary}
\begin{proof} Simply take $\mathcal{K}=1$ and $\chi=\theta$ in the above lemma which is allowed as noticed in Remark~\ref{rm:LimitReg}. 
\end{proof}

We will use Lemma~\ref{lm:abskinloc} for the function
$\mathcal{K}(p)=[|p|^2-s^{-2}]_+$, where $s>0$ is a sufficiently small constant. Here $u_+=\max\{u,0\}$ denotes the positive part of $u$.
\begin{lemma}\label{lm:kinloc-2}
	There exist constants
	$C>0$ and $s^\ast>0$ (depending on the choice of $\chi$) such that for $0<s\leq s^{\ast}$ and any $\ell>0$ we have the inequality for all $\varphi \in H^1(\R^3)$
	\begin{equation}
	\langle \varphi,F_s(|-i\nabla|)\varphi \rangle\geq \int\langle \varphi, Q_u\chi_u(-\Delta-(s\ell)^{-2})_+\chi_u Q_u \varphi\rangle \,d u,
	\end{equation}
	where
	\begin{equation}\label{eq:Fs-2}
	F_s(|p|)=\left\{\begin{array}{lr}
	(|p|^2-\frac12(s\ell )^{-2}),&\hbox{if }|p|\geq \frac56 (s\ell )^{-1},\\
	Cs p^2,&\hbox{if }|p|<\frac56 (s \ell )^{-1}.
	\end{array}\right. 
	\end{equation}
\end{lemma}
\begin{proof} 
By scaling we may assume $\ell=1$. We use (\ref{eq:abskinloc}) and (\ref{eq:absF}) with $\mathcal{K}(p)=(|p|^2-s^{-2})_+$. Since we have chosen $\chi$ to be a Schwartz function, we have $\int \abs{p}^2|\widehat\chi|^2$ being finite and that $\| \mathcal{K}\widehat{\chi}\|_2\leq C_Ns^N$.
	For the first term in (\ref{eq:absF}) we find
	\begin{align*}
	(2\pi)^{-3}\mathcal{K}*\widehat\chi^2(p)&=(2\pi)^{-3}\int (|p-q|^2-s^{-2})\widehat\chi^2(q)\,d q
	+ (2\pi)^{-3}\int [ s^{-2} - p^2 +2 p q - q^2 ]_{+}  \widehat\chi^2(q)\,d q
	\\
	&\leq p^2 - s^{-2} + [s^{-2} - \frac{6}{5} p^2]_{+} +\frac{5-6}{6}  (2\pi)^{-3} \int q^2 \widehat\chi^2(q)\,d q \\
	&\leq p^2 - s^{-2} + [s^{-2} - \frac{6}{5} p^2]_{+},
	\end{align*}
where we used that $t \mapsto [t]_{+}$ is increasing and $[a+b]_{+} \leq [a]_{+} + [b]_{+}$. If $|p|\geq\frac56s^{-1}$ we thus find
	\begin{eqnarray}
(2\pi)^{-3}\mathcal{K}*\widehat\chi^2(p)\leq
	p^2 - s^{-2} + \frac{1}{6}s^{-2}\leq ( p^2 - \frac{1}{2} s^{-2}) - \frac{1}{3} s^{-2}.
	\label{eq:Ffirst-2}
	\end{eqnarray}
	For the second term in (\ref{eq:absF}) we find since
	$\widehat\theta\leq 1$ that 
	\begin{eqnarray}
	|\widehat\theta(p)\widehat\chi*(\mathcal{K}\widehat\chi)(p)|\leq
	\|\widehat{\chi}\|_2\|\mathcal{K}\widehat{\chi}\|_2\leq
	C.\label{eq:F2nd-2}
	\end{eqnarray}
	For the third term in (\ref{eq:absF}) we have similarly
	\begin{eqnarray}
	|\widehat\theta(p)|^2\int \mathcal{K}|\widehat\chi|^2\leq \int|q|^2\widehat\chi(q)^2\,d q\leq C.
	\label{eq:F3rd-2}
	\end{eqnarray}
	For $|p|\geq \frac56 s^{-1}$ we therefore have that the
	function $F$ in (\ref{eq:absF}) satisfies
	$$
	F(p)\leq (|p|^2-\frac12s^{-2})-\frac1{3}s^{-2} +C.
	$$
	With $s^\ast$ sufficiently small we arrive at 
	the first line in (\ref{eq:Fs-2}).\\
	\indent We turn to the proof of the second line in (\ref{eq:Fs-2}). We know that $F(0)=\nabla F(0)=0$. The lemma follows from Taylor's formula if we can show that 
	for $|p|<\frac56s^{-1}$, we have
	\begin{equation}\label{eq:F2der-2}
	|\partial_i\partial_jF(p)|\leq Cs.
	\end{equation} 	
(Actually, the same proof gives $|\partial_i\partial_jF(p)|\leq C_N s^N$ for any power $N$, but we do not need this.)
For the first term in (\ref{eq:absF}) we therefore find for $|p|<\frac56s^{-1}$,
\begin{align*}
|\partial_i\partial_j (\mathcal{K}*\widehat\chi^2)(p)|
&=\bigg |\int (\abs{p-q}^2-s^{-2})_+\partial_i\partial_j\widehat\chi^2(q) \, dq\bigg |\nonumber\\
&\leq C\int_{\{\abs{q}\geq (6s)^{-1}\}}(s^{-2}+\abs{q}^2)\big |\partial_i\partial_j\widehat\chi^2(q)\big | \, dq\nonumber\\
&\leq Cs,
\end{align*}
where we used the fast decay of $\widehat{\chi}$ to conclude.\\

	For the second and third term in (\ref{eq:absF}) we use the fact that for all $i,j=1,2,3$ the numbers
	$$
	\|\widehat\theta\|_\infty,\
	\|\partial_i\widehat\theta\|_\infty,\ 
	\|\partial_i\partial_j\widehat\theta\|_\infty,\ 
	\int|\widehat\chi|^2,\
	\int|\partial_i\widehat\chi|^2,\
	\int|\partial_i\partial_j\widehat\chi|^2
	$$
	are bounded by a constant. The same estimates that led to (\ref{eq:F2nd-2}) and (\ref{eq:F3rd-2})
	then imply \eqref{eq:F2der-2}.
\end{proof}

\begin{lemma}\label{lem:LocKinEn-2}
	There exists a universal constant $b>0$ such that if $s$ is small enough, then for all $\varphi \in H^1_0(\Lambda)$ and all $\ell>0$
	$$
	\langle \varphi, - \Delta \phi\rangle \geq 
	\int_{\ell^{-1}(\Lambda+B(0,\ell/2))} \left\langle \phi, Q_u\left[
	\chi_u \mathcal{K}(-i \nabla) \chi_u
	+
	b \ell^{-2}   \right] Q_u \varphi \right\rangle\,  du,
	$$
	with $\mathcal{K}(p) = (|p|^2-\frac{1}{4} (s\ell)^{-2})_{+}$.
\end{lemma}
\begin{proof} We again consider $\ell=1$.
	By Corollary~\ref{eq:Quav} and a Taylor expansion at $p=0$, we have 
	\begin{equation}\label{eq:QuAv-2}
	\int Q_u \,d u\leq \beta^{-1}\frac{-\Delta}{-\Delta+\beta}
	\end{equation}
	for a universal constant $0<\beta<1$. We use Lemma~\ref{lm:kinloc-2} with $s$ replaced by $2s$. We then find
	\begin{align*}
	\int_{\R^3}Q_u\chi_u(-\Delta-\frac{1}{4} s^{-2})_+\chi_uQ_u \,d u
	+b\int_{\R^3} Q_u\,d u\leq F_{2s}(|-i\nabla|)+b\beta^{-1}\frac{-\Delta}{-\Delta+\beta}.
	\end{align*}
	For $|p|< (5/12)s^{-1}$ and $s, b$ sufficiently small we get
	$$
	F_{2s}(p)+b\beta^{-1}\frac{p^2}{p^2+\beta}
	\leq C sp^2+b\beta^{-1}\frac{p^2}{p^2+\beta}\leq 
	(Cs+b\beta^{-2})p^2\leq p^2.
	$$
	For $|p|\geq(5/12)s^{-1}$ and $s, b$ sufficiently small we get
	\begin{align*}
F_{2s}(p)+b\beta^{-1}\frac{p^2}{p^2+\beta}=
	(p^2-\frac{1}{8}s^{-2})+b\beta^{-1}\frac{p^2}{p^2+\beta}
	\leq p^2-\frac{1}{8}s^{-2}+b\beta^{-1}\leq p^2.
	\end{align*}
\end{proof}

\subsection{The localized Hamiltonian}

Let $\mathcal{T}_u$ be the localized kinetic energy operator, as defined in \eqref{eq:DefTu}, $\rmu$ such that \eqref{eq:R/ell} is satisfied and define for $(x_1,\ldots, x_N) \in {\R^{3N}}$,
\begin{align}
{\mathcal W}_u(x_1,\ldots,x_N) :=
\sum_{i=1}^N - \rmu \int w_{1,u}(x_i,y)\,dy + \sum_{i<j} w_u(x_i,x_j).
\end{align}
We also abbreviate
\begin{align}
\mathcal{T}:= \mathcal{T}_{u=0},\qquad {\mathcal W}(x_1,\ldots,x_N) := {\mathcal W}_{u=0}(x_1,\ldots,x_N).
\end{align}
Define the operator ${\mathcal H}_{B,u}(\rmu)$ on the symmetric Fock space over $L^2({\R^3})\supset L^2({\Lambda})$, to preserve particle number and satisfy that
\begin{align}
({\mathcal H}_{B,u}(\rmu))_{N} = \sum_{i=1}^N \mathcal{T}_{u,i} + {\mathcal W}_u(x_1,\ldots,x_N).
\end{align}
As above we abbreviate
$$
{\mathcal H}_{B}(\rmu):= {\mathcal H}_{B,u=0}(\rmu). 
%\qquad H^N_{B}(\rmu):=H^N_{B,u=0}(\rmu).
$$
We will also write
$$
\chi_B := \chi_{u=0} = \chi(\,\cdot\,/\ell).
$$
Define the \emph{box energy} and \emph{box energy density}, by
\begin{align}
E_B(\rmu) &:= \inf \Spec {\mathcal H}_{B}(\rmu), \\
e_B(\rmu) &:= \ell^{-3} \inf \Spec {\mathcal H}_{B}(\rmu) = \ell^{-3} E_B(\rmu).
\end{align}
Notice that $E_B(\rmu), e_B(\rmu)$ depend on the localization function $\chi$---since $\mathcal{T}$ and ${\mathcal W}$ do---but we choose not to let the notation reflect this dependence.
With these conventions, we find

\begin{theorem}\label{thm:CompareBoxEnergy}
If $\rho_\mu$ is sufficiently small so that  \eqref{eq:R/ell} is satisfied, then we have
\begin{align}
e_0(\rmu) \geq e_B(\rmu).
\end{align}
\end{theorem}

\begin{proof}
Note that $({\mathcal H}_{B,u}(\rmu))_{N}$ and $({\mathcal H}_{B,u'}(\rmu))_{N}$ are unitarily equivalent by \eqref{eq:transInv}.

From Lemma~\ref{lem:LocPotential} and Lemma~\ref{lem:LocKinEn-2} we find that
\begin{align}
{\mathcal H}_{\rmu,N}(\rmu) \geq
\int_{\ell^{-1}(\Lambda + B(0,\ell/2))} ({\mathcal H}_{B,u}(\rmu))_{N} \,du  \geq \ell^{-3} | \Lambda + B(0,\ell/2)| E_{B}(\rmu).
\end{align}
Now the desired result follows upon using that $|\Lambda+B(0,\ell/2)|/|\Lambda| \rightarrow 1$ in the thermodynamic limit.
\end{proof}
\section{Energy in the box}\label{sec:BoxEnergy}
It is clear, using Theorem~\ref{thm:CompareBoxEnergy}, that Theorem~\ref{thm:LHY-Background} is a consequence of the following theorem on the box Hamiltonian.
\begin{theorem}\label{thm:LHY-Box}
For a given localization function $\chi$, there exist universal constants $K_0, C'>0$ so that the following is true. Suppose $v$ satisfies Assumption~\ref{assump:v} and choose $K = K_0$ for the parameter appearing in the definition of $\ell$ in \eqref{eq:def_ell}.
If $R/\ell\leq C'$ and \eqref{eq:R/ell} is satisfied, then the box ground state energy density, $e_B(\rmu)$, satisfies the bound
	\begin{align}\label{eq:EnergyBoxRes1}
	e_B(\rmu) \geq -4\pi \rmu^2 a- C\rho_\mu^2a (\rmu a^3)^{1/2} - C \rho_\mu^3 a^2 R^2 .
	\end{align}
\end{theorem}

\begin{proof}[Proof of Theorem~\ref{thm:LHY-Background} ]
We will show that Theorem~\ref{thm:LHY-Background} follows from Theorem~\ref{thm:CompareBoxEnergy} and Theorem~\ref{thm:LHY-Box}.
Choose and fix a localization function $\chi \in C^{\infty}_c({\mathbb R}^3)$ as in \eqref{def: chi}.
We take $K=K_0$ in the definition of $\ell$ and use Theorem~\ref{thm:LHY-Box} for this choice. Clearly \eqref{eq:R/ell} is automatically satisfied if $C'\leq D$.
Now Theorem~\ref{thm:LHY-Background} follows using Theorem~\ref{thm:CompareBoxEnergy}.
\end{proof}

The remainder of this paper will be dedicated to collecting the ingredients that we need for the proof of Theorem~\ref{thm:LHY-Box}, which will be given 
in Section~\ref{sec:TotalEnergy}.
\subsection{Particle numbers and densities}
Recall the projections $P_u, Q_u$ defined in \eqref{def: projections}. Since now we are working on a fixed box $B=[-\frac{\ell}{2},\frac{\ell}{2}]^3$, we will just denote them by $P$ and $Q$.
Notice that $P+Q=:I_B$ is the orthogonal projection, in $L^2({\mathbb R}^3)$ onto the subspace of functions supported in $B$.

Define the operators
$$
n := \sum_{i=1}^N I_{B,i},\qquad n_0 := \sum_{i=1}^N P_i,\qquad n_{+} := \sum_{i=1}^N Q_i = n- n_0.
$$

Because the operator $n$ commutes with ${\mathcal H}_{B}(\rmu)$, we can restrict to eigenspaces of $n$ and therefore simultaneously treat $n$ as an operator and a parameter.

Recall, that $\rmu$ is the parameter introduced in \eqref{eq:BackgroundH}. We define
\begin{align}\label{eq:Densities}
\rho:= n \ell^{-3},\qquad \rho_{+} := n_{+} \ell^{-3}, \qquad \rho_0:= n_{0} \ell^{-3}.
\end{align}
\subsection{Estimates on non-quadratic terms}
We can, for each $j$, write $I_{j}  = P_j + Q_j$. Inserting this on both sides of our operator and expanding we will get a number of terms. These we will organize depending on the number of $Q$'s involved. For an even finer decomposition of some of the terms we invoke the scattering solution, $\omega$. The leading order term in \eqref{eq:EnergyBoxRes1} will be obtained using the Bogoliubov diagonalization carried out in Section~\ref{sec:Bogoliubov}. This diagonalization involves the terms quadratic in $Q$ from the localized kinetic energy and most of the terms quadratic in $Q$ from the localized potential energy.
The aim here is to estimate the non-quadratic terms.

\begin{lemma}[Potential energy decomposition]\label{lem:SplittingInQs}
	We have for $\rmu$ such that \eqref{eq:R/ell} is satisfied
	\begin{equation} \label{eq:potsplit}
		-\rmu \sum_{i=1}^N \int w_1(x_i,y)\,dy+
		\frac{1}{2} \sum_{i\neq j}  w(x_i, x_j)
		= {\mathcal Q}_0^{\rm ren}+{\mathcal Q}_1^{\rm ren}
		+{\mathcal Q}_2^{\rm ren}+{\mathcal Q}_3^{\rm ren} + {\mathcal Q}_4^{\rm ren},
	\end{equation}
	where
	\begin{align}
	{\mathcal Q}_4^{\rm ren}:=&\,
	\frac{1}{2} \sum_{i\neq j} \Big[ Q_i Q_j + (P_i P_j + P_i Q_j + Q_i P_j)\omega(x_i-x_j) \Big] w(x_i,x_j) \nonumber \\
	&\,\qquad \qquad \times
	\Big[ Q_j Q_i + \omega(x_i-x_j) (P_j P_i + P_j Q_i + Q_j P_i)\Big],\label{eq:DefQ4}\\
	{\mathcal Q}_3^{\rm ren}:=&\,
	\sum_{i\neq j} P_i Q_j w_1(x_i,x_j) Q_j Q_i + h.c., \label{eq:DefQ3} \\
	{\mathcal Q}_2^{\rm ren}:=&\, \sum_{i\neq j} P_i Q_j w_2(x_i,x_j) P_j Q_i
	+ \sum_{i\neq j} P_i Q_j w_2(x_i,x_j) Q_j P_i \nonumber \\
	&- \rmu \sum_{i=1}^N Q_i \int w_1(x_i,y)\,dy Q_i+\frac{1}{2}\sum_{i\neq j} (P_i P_j w_1(x_i,x_j) Q_j Q_i + h.c.) 	,\label{eq:DefQ2}\\
	{\mathcal Q}_1^{\rm ren}:=&\, \sum_{i\neq j}P_jQ_iw_2(x_i,x_j)P_iP_j-\rmu
	\sum_{i} Q_i \int w_1(x_i,y)\,dy P_i +h.c.,
	\label{eq:DefQ1} \\
	{\mathcal Q}_0^{\rm ren}:=&\,
	\frac{1}{2} \sum_{i\neq j} P_i P_j w_2(x_i,x_j) P_j P_i - \rmu \sum_i P_i \int w_1(x_i,y)\,dy P_i.\label{eq:DefQ0}
	\end{align}
\end{lemma}

\begin{proof}
	The identity \eqref{eq:potsplit} follows from simple algebra using the identities $P_i+Q_i = 1_{i}$, $w_1=w(1-\omega)$ and $w_2=w(1-\omega^2)$. In fact it is easy to see that adding the terms in which $\rmu$ appears gives the one-body term in \eqref{eq:potsplit}. Regarding the two-body term in \eqref{eq:potsplit} we first insert $1_{i}=P_i+Q_i $ for all $i$ on both sides of $w(x_i,x_j)$ and organize the $16$ resulting terms by the number of $Q$'s occurring. Then, if three or less $Q$'s occur, we replace $w$ by either $w_1$ or $w_2$ and add a corresponding term to the $4$-$Q$ term.
\end{proof}%

Applying the decomposition of the potential energy in Lemma~\ref{lem:SplittingInQs} we arrive at the following lemma by, in particular, applying a Cauchy-Schwarz inequality to absorb ${\mathcal Q}_3^{\rm ren}$ in the positive ${\mathcal Q}_4^{\rm ren}$-term.
\begin{lemma}\label{lm:appinteractionestimate} There is a constant $C>0$, depending only on the localization function $\chi$, such that if $\rmu$ satisfies \eqref{eq:R/ell}, then
	\begin{align}\label{eq:SmallsimpleQs}
	-\rmu \sum_{i=1}^N \int w_{1}(x,y)\,dy+
	\frac{1}{2} \sum_{i\neq j}  w(x_i, x_j)
	\geq A_0+A_2-Ca (\rmu +n_0|B|^{-1})n_+  
	\end{align}
	where
	\begin{align}
	A_0={}&\frac{n_0(n_0-1)}{2|B|}\big(\widehat{g}(0) + \widehat{g\omega}(0)\big)
	- \left(\rmu \frac{n_0}{|B|}+\frac14\left(\rmu
	-\frac{n_0+1}{|B|}\right)^2\right)|B| \widehat{g}(0)
	\label{eq:A0}
	\end{align}
	and 
	\begin{equation}\label{eq:defA_2}
		A_2= \frac{1}{2}\sum_{i\neq j} P_{i} P_{j} w_{1}(x_i,x_j) Q_{j} Q_{i} + h.c.
	\end{equation} 
\end{lemma}
\begin{proof}
We use the identity \eqref{eq:potsplit} and note that, since $P$ is the projection onto constant functions in the box,
\begin{align}\label{eq:Q_0^ren equation}
	{\mathcal Q}_{0}^{\rm ren}&=\frac{n_0(n_0-1)}{2|B|^2}\iint w_{2}(x,y)\,dx dy - \rmu \frac{n_0}{|B|}\iint w_{1}(x,y)\,dx dy \nonumber \\
	&=\frac{n_0(n_0-1)}{2|B|} \big(\widehat{g}(0) + \widehat{g\omega}(0)\big)- 
	\rmu n_0 \widehat{g}(0),
\end{align}
where we used \eqref{eq:DefU} and \eqref{eq:w2int} to get the last identity.

We will now show that
\begin{align}\label{eq:Absorb}
{\mathcal Q}_{1}^{\rm ren}+{\mathcal Q}_{3}^{\rm ren}+{\mathcal Q}_{4}^{\rm ren} \geq &\,-\frac14\left(\rmu-\frac{n_0+1}{|B|}\right)^2\iint w_{1}(x,y)\,dx dy\nonumber\\
&\;-Can_0\abs{B}^{-1}n_+-Ca\rmu n_+.
\end{align}
Combining \eqref{eq:Q_0^ren equation} with \eqref{eq:Absorb} and again using \eqref{eq:DefU} we easily get
\begin{align}
{\mathcal Q}_{0}^{\rm ren}+{\mathcal Q}_{1}^{\rm ren}+{\mathcal Q}_{3}^{\rm ren}+{\mathcal Q}_{4}^{\rm ren}\geq A_0-Ca (\rmu +n_0|B|^{-1})n_+ .
\end{align}

We have, using Lemma~\ref{lem:IntEst},
$$
0\leq \sum_{i,j}P_iQ_jw_{1}(x_i,x_j)Q_jP_i=n_0|B|^{-1}\sum_j Q_j\chi_B(x_j)W_1*\chi_B(x_j)Q_j
\leq Cn_0 n_+\ell^{-3}a\|\chi_B\|_\infty^2
$$
or more generally using again Cauchy-Schwarz inequalities and Lemma~\ref{lem:IntEst}, we have for all $k\in {\mathbb N} \cup \{ 0 \}$
\begin{align}
0\leq \sum_{i,j}P_iQ_j(w_{1}\omega^k)(x_i,x_j)Q_jP_i\leq \,& Cn_0\ell^{-3}a\|\chi_B\|_\infty^2n_+,\label{eq:Q'app1}
\end{align}
\begin{align}
 \pm\Bigl(\sum_{i,j}P_iQ_j(w_{1}\omega^k)(x_i,x_j)P_jQ_i+h.c.\Bigr)\leq&\,
 2\sum_{i,j}P_iQ_j(w_{1}\omega^k)(x_i,x_j)Q_jP_i\nonumber\\
\leq &\, Cn_0\ell^{-3}a\|\chi_B\|_\infty^2n_+,\label{eq:Q'app2}\\
\pm\Bigl(\sum_{i,j}Q_iP_j(w_1\omega^k)(x_i,x_j)P_jP_i+h.c.\Bigr)\leq &\,
\sum_{i,j}Q_iP_j(w_1\omega^k)(x_i,x_j)P_jQ_i\nonumber\\
&\,+\sum_{i,j}P_iP_j(w_1\omega^k)(x_i,x_j)P_jP_i
\nonumber\\ \leq &\,C
n_0a\ell^{-3}\Bigl(\|\chi_B\|_\infty^2n_+ +n_0\Bigr),\label{eq:Q'app}
\end{align}
where we have abbreviated $(w_1\omega^k)(x_1,x_2)=w_1(x_1,x_2)\omega(x_1-x_2)^k$. We have
	\begin{align}
	\sum_{i,j}P_iQ_jw_1(x_i,x_j)Q_jQ_i=&\, \sum_{i,j} \Big(P_i
	Q_j w_1(x_i,x_j) \Big[ Q_j Q_i + \omega(x_i-x_j) (P_j P_i + P_j Q_i
	+ Q_j P_i)\Big]  \Big)\nonumber \\ &\,- \sum_{i,j}
	\Big(P_i Q_j w_1(x_i,x_j) \omega(x_i-x_j) (P_j P_i +
	P_j Q_i + Q_j P_i)\Big) \label{eq:3Qto4Q}
	\end{align}
	and the same identity for the Hermitian conjugates.	We estimate the first term in \eqref{eq:3Qto4Q} (and its Hermitian conjugate) using a Cauchy-Schwarz inequality
	\begin{align}
	\pm \sum_{i,j} \Big(P_i
	&Q_j w_1(x_i,x_j) \Big[ Q_j Q_i + \omega(x_i-x_j) (P_j P_i + P_j Q_i
	+ Q_j P_i)\Big]+h.c.  \Big) \nonumber \\
	&\,\leq \frac{1}{2}
	{\mathcal Q}_4^{\rm ren}+C\sum_{i\neq j} P_i Q_j
	w(x_i,x_j)(1-\omega(x_i-x_j))^2Q_j P_i \nonumber \\
	&\,\leq \frac{1}{2}
	{\mathcal Q}_4^{\rm ren}+C\sum_{i\neq j} P_i Q_j
	w_1(x_i,x_j)Q_j P_i,
	\end{align}
	where we have used the pointwise inequality $0 \leq \omega \leq 1$ in the last inequality.\\
	We estimate the second term in \eqref{eq:3Qto4Q} (and its Hermitian conjugate) using a Cauchy-Schwarz inequality
\begin{align}\label{eq:CSexample}
&-\sum_{i,j}
\Big(P_i Q_j w_1(x_i,x_j) \omega(x_i-x_j) (P_j P_i +
P_j Q_i + Q_j P_i)+h.c.\Big)  \\
&\,\geq -\sum_{i,j}
\Big(P_i Q_j w_1(x_i,x_j) \omega(x_i-x_j) P_j P_i +h.c.\Big) 
-4\sum_{i,j}\Big(P_i Q_j w_1(x_i,x_j) \omega(x_i-x_j) Q_j P_i \Big).\nonumber
	\end{align}
	Thus applying a Cauchy-Schwarz inequality, that ${\mathcal Q}_4^{\rm ren}\geq 0$ and the estimates \eqref{eq:Q'app1}-\eqref{eq:Q'app2} we arrive at
\begin{align}\label{eq:Q3+Q4CS}
	{\mathcal Q}_{3}^{\rm ren}+{\mathcal Q}_{4}^{\rm ren} &\,\geq -C\sum_{i\neq j} P_{i} Q_{j} w_{1}(x_i,x_j) Q_{j} P_{i}
	- \sum_{i\neq j} \Big( P_{i} Q_{j} w_{1} \omega(x_i,x_j) P_{j} P_{i} +
	h.c.\Big) \nonumber 
	\\ &\,\quad - 4 \sum_{i \neq j} P_{i} Q_{j}
	w_{1} \omega (x_i,x_j) Q_{j} P_{i} \nonumber \\
	&\,\geq - \sum_{i\neq j} \Big(
	P_{i} Q_{j} w_{1} \omega(x_i,x_j) P_{j} P_{i} + h.c.\Big)-Can_0|B|^{-1}n_+\|\chi_B\|_\infty^2 .
	\end{align}

	Notice that if we rewrite ${\mathcal Q}_{1}^{\rm ren}$ as
	\begin{align}  
	{\mathcal Q}_1^{\rm ren}= &\, (n_0|B|^{-1} -\rmu) \sum_{i} Q_i \chi_{B}(x_i) W_1*\chi_{B}(x_i)  P_i + h.c. 
	\nonumber\\ 
	&\,+ n_0|B|^{-1}  \sum_{i} Q_i \chi_{B}(x_i)  (W_1\omega)*\chi_{B}(x_i) P_i + h.c.,
	\label{eq:Q1n0}
	\end{align}
	then the first term on the right side of \eqref{eq:Q3+Q4CS} cancels the second line of \eqref{eq:Q1n0}.
	
	Since $\sum_i P_i A_i Q_i n_{+} = (n_{+}+1) \sum_i P_i A_i Q_i$ (for any bounded, self-adjoint one-particle operator $A$), we have
	$$
	\sum_i P_i A_i Q_i p(n_{+}) = p(n_{+}+1) \sum_i P_i A_i Q_i,
	$$
	for any polynomial $p$. Hence, by a limiting argument,
	\begin{align}
	\sum_i P_i A_i Q_i n_{+} = \sum_i P_i A_i Q_i \sqrt{n_{+} } \sqrt{n_{+} } 
	= \sqrt{n_{+} +1} 
	\sum_i P_i A_i Q_i \sqrt{n_{+} }.\label{eq: sqr n_+}
	\end{align}
	With \eqref{eq: sqr n_+} at hand we estimate the remaining part of ${\mathcal Q}_{1}^{\rm ren}$
	\begin{align}
	|B&|^{-1}(n_0-\rmu|B|)\sum_{i} Q_{i} \chi_B(x_i) W_1*\chi_B(x_i)
	P_{i} +
	h.c.\nonumber\\ =\,&|B|^{-1}({n_0}^{1/2}+(\rmu|B|)^{1/2})\sum_{i}
	Q_{i} \chi_B(x_i) W_1*\chi_B(x_i)
	P_{i}((n_0+1)^{1/2}-(\rmu|B|)^{1/2}) +
	h.c.\nonumber\\ \geq\,&-4|B|^{-1}\left(n_0^{1/2}+(\rmu|B|)^{1/2}\right)^2\sum_{i}
	Q_{i} \chi_B(x_i) W_1*\chi_B(x_i) Q_{i}\nonumber\\&
	-\frac14|B|^{-1}\left((n_0+1)^{1/2}-(\rmu|B|)^{1/2}\right)^2\sum_{i}
	P_{i} \chi_B(x_i) W_1*\chi_B(x_i) P_{i}.\label{Q_1 remaining part}
	\end{align}
	The first term above we estimate as
	\begin{align}
	&-4|B|^{-1}\left(n_0^{1/2}+(\rmu|B|)^{1/2}\right)^2\sum_{i}
	Q_{i} \chi_B(x_i) W_1*\chi_B(x_i) Q_{i}\\
	&\quad \geq -C|B|^{-1}(n_0+\rmu|B|)n_+ a\|\chi_B\|_\infty^2.
	\end{align}
We complete the proof of \eqref{eq:Absorb} by estimating the last term in \eqref{Q_1 remaining part}
	\begin{align}
	&-\frac14\frac{n_0}{|B|^2}\left((n_0+1)^{1/2}-(\rmu|B|)^{1/2}\right)^2\iint w_{1}(x,y)\,dx dy\nonumber\\
	&=\,-\frac{1}{4}\frac{n_0}{|B|^2}\left((n_0+1)-\rmu|B|\right)^2[(n_0+1)^{1/2}+(\rmu|B|)^{1/2}]^{-2}\iint w_{1}(x,y)\,dx dy\nonumber\\
	&\geq\,-\frac{1}{4}\left(\frac{n_0+1}{|B|}-\rmu\right)^2\iint w_{1}(x,y)\,dx dy,
	\end{align}
	which together with ${\mathcal Q}^{\rm ren}_{0}$ gives the $A_0$ term in the lemma. 
	
	Recalling that $w_2=w_1(1+\omega)\leq 2w_1$ we absorb the first two terms in ${\mathcal Q}_{2}^{\rm ren}$ into the last term in \eqref{eq:SmallsimpleQs} 
	using again the same Cauchy-Schwarz as in the second inequality in \eqref{eq:CSexample}. Finally, the one-body term in ${\mathcal Q}_{2}^{\rm ren}$ is estimated as
\begin{align}
\rmu \sum_{i} Q_i \int w_1(x_i,y)\,dy\, Q_i +h.c.\leq Ca\rmu n_+\|\chi_B\|_\infty^2.
\end{align}%
\end{proof}%
\section{Bogoliubov calculation}\label{sec:Bogoliubov}
In this section, we will study the `effective Bogoliubov' Hamiltonian, i.e. the remaining terms quadratic in $Q$. 
We will assume that the number of particles $n$ satisfies $n \leq M_0 \rmu \ell^3$, where $M_{0}$ is some fixed constant.
In order to control the number of exited particles, $n_{+}$, we separate the `gap' from the kinetic energy in \eqref{eq:DefTu}, i.e. the constant term $b \ell^{-2} Q$. This {\it positive} term will be very important later---see the proof of Lemma~\ref{lem:EnergyFirst} below. That is, we define ${\mathcal H}^{\rm Bog}$ as an operator on the Fock space such that on the $N$-particle sector we have
\begin{align}\label{eq:Def_HBog}
({\mathcal H}^{\rm Bog})_{N} = \sum_{j=1}^N \mathcal{T}_j - b \ell^{-2} n_{+} + A_2%
\end{align}
with $\mathcal{T}$ from \eqref{eq:DefTu} and $A_2$ from \eqref{eq:defA_2}.

We will pass to a second quantized formalism in order to give an effective lower bound to this operator. 
We define $a_0$ as the annihilation operator associated to the condensate function for the box $B$, i.e. for $\Psi \in \otimes_{s}^N L^2$ we have
$$
(a_0 \Psi)(x_2, \ldots, x_N) := \frac{\sqrt{N}}{\ell^{3/2}} \int \theta(y) \Psi(y,x_2,\ldots, x_N) \,dy.
$$
Therefore,
\begin{align}\label{eq:Appa0}
\langle \Psi ,n_0\Psi\rangle=\langle \Psi\,|\, a_0^{*} a_0 \Psi \rangle 
=
\frac{N}{\ell^3} \int \left| \int \theta(y) \Psi(y,x_2,\ldots, x_N) \,dy\right|^2 dx_2\cdots dx_N. 
\end{align}
We define, for $k \in {\mathbb R}^3$,
\begin{align}\label{def: b_k and b^*_k}
b_k := a_0^{*} a( Q(e^{ikx} \chi_B))\qquad \textrm{and}\qquad b_k^{*} := a( Q(e^{ikx} \chi_B))^{*} a_0.
\end{align}
Then,
\begin{align}
[ b_k, b_{k'} ] = 0, \qquad \forall k,k' \in {\mathbb R}^3,
\end{align}
and
\begin{align}\label{eq:CommRel}
[b_k, b_{k'}^{*}] =
a_0^{*} a_0 \langle Q(e^{ikx} \chi_B), Q (e^{ik'x} \chi_B) \rangle-
a( Q(e^{ik'x} \chi_B))^{*} a( Q(e^{ikx} \chi_B)).
\end{align}%
In particular,
\begin{align}\label{eq:Commutator}
[b_k, b_{k}^{*}] \leq a_0^{*} a_0 \int \chi_B^2 = \ell^3 a_0^{*} a_0.
\end{align}

By a calculation similar to \eqref{eq:Appa0}, we have
\begin{align}\label{eq:nPlus2nd}
\langle \Psi, (2\pi)^{-3} \int a( Q(e^{ikx} \chi_B))^{*} a( Q(e^{ikx} \chi_B))\,dk \,\Psi \rangle
= \langle \Psi, \sum_i Q_i \chi_B(x_i)^2 Q_i \Psi \rangle.
\end{align}
Therefore, using \eqref{eq:Appa0} and $[a_0, a_0^{*}] = 1$, we find for $\Psi \in \otimes_{s}^N L^2$, 
\begin{align}
\langle \Psi, (2\pi)^{-3} \int b_k^{*} b_k \,dk \Psi \rangle &=
\langle \Psi, (2\pi)^{-3} \int a( Q(e^{ikx} \chi_B))^{*} (n_0+1) a( Q(e^{ikx} \chi_B))\,dk \,\Psi \rangle \nonumber \\
&\leq 
N \langle \Psi, (2\pi)^{-3} \int a( Q(e^{ikx} \chi_B))^{*} a( Q(e^{ikx} \chi_B))\,dk \,\Psi \rangle.
\end{align}
Therefore we get, using \eqref{eq:nPlus2nd},
\begin{align}\label{eq: b^*b integral}
\int b_k^{*}b_k\, dk \leq C nn_+\,,
\end{align}
with $C = \| \chi_B\|_{\infty}^2$.
Furthermore, we introduce the Fourier multiplier corresponding to the localized kinetic energy (after the separation of the constant term), i.e.
\begin{align}\label{def: tau}
\tau(k) := \left(|k|^2 - \constkinetickernel \ell^{-2} \right)_{+}
\end{align}%
allowing us to write
\begin{align}\label{eq:Def_HBog short}
({\mathcal H}^{\rm Bog})_{N} = \sum_{j=1}^N Q_{B,j}
\chi_B 
\tau(-i\nabla)
\chi_B
 Q_{B,j}+ A_2.
\end{align}
\begin{lemma}[Lower bound by second quantized operator]\label{prop: lower bound on H_Bog}~\\
Suppose that \eqref{eq:R/ell} is satisfied. Then, with the notation above, in particular \eqref{eq:Def_HBog}, we have
\begin{align}\label{eq:EstH1-Bog}
{\mathcal H}^{\rm Bog} \geq
{\mathcal H}^{\rm Bog}_1 - C a (\rho+\rmu) n_{+},
\end{align}
where $C$ only depends on the localization function $\chi$ and
with
\begin{align}\label{eq:2ndQuantized}
{\mathcal H}^{\rm Bog}_1:=
\frac{1}{2} (2\pi)^{-3} \int_{{\mathbb R}^3} {\mathcal A}(k) \left( b_k^{*}b_k + b_{-k}^{*} b_{-k} 
\right)
+ {\mathcal B}(k) \left( b_k^{*}b_{-k}^{*} + b_{k} b_{-k} 
\right)\,dk.
\end{align}
Here ${\mathcal A}(k)=0$ if $n=0$ and otherwise we let
\begin{align}\label{eq:AandB}
{\mathcal A}(k) = \frac{1}{n} \big( \tau(k)+ \CA\rho a + \rmu a \big)
\qquad
\text{ and }
\qquad 
{\mathcal B}(k) =\frac{\widehat{W}_{1}(k)}{\ell^3}.
\end{align}
\end{lemma}
\begin{proof}
We can write, with $\theta_0 = \one_{B(0)}$, 
\begin{align}
\langle \varphi \,|\, P e^{ikx} \chi(x/\ell) Q \psi \rangle
&= \ell^{-3} 
\langle \varphi \,|\, \theta_0 \rangle \langle \theta_0 \,|\, e^{ikx} \chi(x/\ell) Q \psi \rangle  \nonumber \\
&=\ell^{-3} \langle \varphi \,|\, \theta_0 \rangle \langle Q(e^{-ikx} \chi(x/\ell)) \,|\, \psi \rangle \nonumber \\
&= \ell^{-3/2} \langle \varphi \,|\, b_{-k} \psi \rangle.
\end{align}
Therefore, we see that the second quantization of $P e^{ikx} \chi(x/\ell) Q$ is $\ell^{-3/2} b_{-k}$.

An application of \eqref{eq: b^*b integral} yields
\begin{align}\label{Eq: 2nd quant B II}
& \frac{1}{2}(2\pi)^{-3}\frac{a}{n} \int (\CA\rho  + \rmu ) \big [b_{k}^{*} b_{k}+b_{-k}^{*} b_{-k} \big ] \,dk
\leq Ca(\rho+\rmu)n_+.
\end{align}%
Furthermore
\begin{align}\label{eq: P_1P2Q_2Q_1}
\sum_{j \neq s} P_j P_s w_1(x_j,x_s) Q_s Q_j
& = \sum_{j \neq s}(2\pi)^{-3} \int \widehat{W}_1(k) (P_j \chi(x_j/\ell) e^{ikx_j} Q_j )
 (P_s \chi(x_s/\ell) e^{-ikx_s} Q_s )\, dk\nonumber\\
 & =  (2\pi)^{-3} \ell^{-3} \int \widehat{W}_1(k) b_{k}b_{-k}\, dk.
\end{align}
The term $Q\chi \tau (-i\nabla)\chi Q$ second quantizes as
\begin{align}
Q\chi \tau (-i\nabla)\chi Q&=(2\pi)^{-3}\int_{\R^3}\tau(k)a(Q(\chi e^{ikx}))^*
a(Q(\chi e^{ikx}))d k\nonumber.
\end{align}
Therefore we estimate $\sum_{j=1}^{N}Q_j\chi\tau(-i\nabla)\chi Q_j$ in terms of its second quantization as
\begin{align}
\sum_{j=1}^{N}Q_j\chi\tau(-i\nabla)\chi Q_j
&\geq  (2\pi)^{-3}\int_{\R^3}\tau(k)a(Q(\chi e^{ikx}))^*\frac{a_0a_0^*}{n}
a(Q(\chi e^{ikx}))d k\nonumber\\
&=(2\pi)^{-3}\int_{\R^3}\frac{\tau(k)}{n}b_k^*b_kd k\nonumber\\
&=\frac{1}{2} (2\pi)^{-3} \int \frac{\tau(k)}{n} \left[
b^{*}_{k}b_{k} + b^{*}_{-k}b_{-k} 
\right]\, dk.\label{eq:qchitauchiq}
\end{align}
Here we used that $b_k\psi=0$ if $\psi$ is in the condensate allowing us to assume \label{nplus assumption}that $n_+\geq 1$ such that in fact $a_0a_0^\ast\leq n$.
This finishes the proof of Lemma~\ref{prop: lower bound on H_Bog}.

\end{proof}

\begin{lemma}[The Bogoliubov integral]\label{prop:BogIntegral-2-New}~\\
Assume that the number of particles $n$ satisfies the bound  
$n \leq  M_{0} \rmu \ell^3$, where $M_{0}$ is some fixed constant. 
Then, for $\rmu$ sufficiently small so that \eqref{eq:R/ell} is satisfied, we have 
\begin{align}\label{eq:Est-H1Bog_integral}
{\mathcal H}_1^{\rm Bog} \geq
- \frac{1}{2} \ell^3 \rho_0 \rho  \widehat{g\omega}(0) 
-C_1   \ell^3 \rho_0 \rho a  \left( \rmu a^3\right)^{1/2} -C_2   \ell^3 \rho_0 \rho a \frac{R^2}{\ell^2}.
\end{align}
Here $C_1$ depends on $M_0$, on the constant $K$ from \eqref{eq:def_ell} and on the localization function $\chi$, while $C_2$ only depends on $K$ and $\chi$.
\end{lemma}
\begin{proof}
Recall that in \eqref{eq:2ndQuantized} we defined
\begin{align}
{\mathcal H}^{\rm Bog}_1:=
\frac{1}{2} (2\pi)^{-3} \int_{{\mathbb R}^3} {\mathcal A}(k) \left( b_k^{*}b_k + b_{-k}^{*} b_{-k} 
\right)
+ {\mathcal B}(k) \left( b_k^{*}b_{-k}^{*} + b_{k} b_{-k} 
\right)\,dk,
\end{align}
where ${\mathcal A}(k)=0$ if $n=0$ and otherwise we have
\begin{align}\label{eq:AandB}
{\mathcal A}(k) = \frac{\tau(k)}{n} + a \frac{\CA+ \rmu/\rho}{\ell^3}
\qquad
\text{ and }
\qquad 
{\mathcal B}(k) =\frac{\widehat{W}_{1}(k)}{\ell^3}.
\end{align}

By assumption $\rho\leq M_{0} \rho_{\mu}$.
We seek to apply the Bogoliubov method to estimate the quadratic Hamiltonian, see Theorem~\ref{thm:bogolubov}. In order to do so, we need to verify that the condition $-{\mathcal A} < {\mathcal B} \leq {\mathcal A}$ with the notation from \eqref{eq:AandB} is satisfied.
However, since $|\widehat{W}_1(k)|\leq \widehat{W}_1(0) = 8\pi a$, this is trivial.

Notice for later use, that for all $\rmu$ we have actually proved the bounds
\begin{align}\label{eq:UniformBoverA}
\frac{|{\mathcal B} |}{ {\mathcal A}} \leq \frac{1}{2} \qquad \text{and}\qquad \mathcal{A}>0.
\end{align}

Therefore, we may apply Theorem~\ref{thm:bogolubov} to bound ${\mathcal H}^{\rm Bog}_1$.
We obtain
\begin{align}
{\mathcal H}^{\rm Bog}_1\geq \frac{1}{4}(2\pi)^{-3}\int_{\R^3}(\sqrt{\mathcal{A}(k)^2-\mathcal{B}(k)^2}-\mathcal{A}(k))([b_k,b_k^\ast]+[b_{-k},b_{-k}^\ast])\,d k.
\end{align}
We insert the bound $[b_k, b_{k}^{*}] \leq \ell^3 a_0^{*} a_0=\ell^3n_0$ from \eqref{eq:Commutator} and get
\begin{align}
{\mathcal H}^{\rm Bog}_1&\geq \frac{1}{2}(2\pi)^{-3}\ell^3 n_0\int_{\R^3} \sqrt{\mathcal{A}(k)^2-\mathcal{B}(k)^2}-\mathcal{A}(k)\,d k.
\end{align}
Notice that, using \eqref{eq:UniformBoverA}, there exists $C>0$, such that
\begin{align}
\left| \sqrt{\mathcal{A}(k)^2-\mathcal{B}(k)^2}-\mathcal{A}(k) + \frac{\mathcal{B}(k)^2}{2\mathcal{A}(k)} \right| \leq C \frac{\mathcal{B}(k)^4}{\mathcal{A}(k)^3}.\label{Second Born term added}
\end{align}
Therefore, \eqref{eq:Est-H1Bog_integral} follows, using Lemma~\ref{lem:I2-integral} below, if we can prove that
\begin{align}
\int \left| \ell^{3} \rho^{-1} \frac{\mathcal{B}(k)^2}{2\mathcal{A}(k)} -  \frac{\widehat{W}_1(k)^2}{2k^2}\right|\,dk &\leq C a\sqrt{\rmu a^3},\label{eq:BoundsOnIntegralsI} \\
\ell^{3} \rho^{-1} \int \frac{\mathcal{B}(k)^4}{\mathcal{A}^3(k)}\,dk &\leq C a \sqrt{\rmu a^3},\label{eq:BoundsOnIntegralsII}
\end{align}
for some constant $C$ independent of $\rho$ and $\rmu$.

We first prove the bound in \eqref{eq:BoundsOnIntegralsI}. Notice that
\begin{align}
\int \left| \ell^{3} \rho^{-1} \frac{\mathcal{B}(k)^2}{2\mathcal{A}(k)} -  \frac{\widehat{W}_1(k)^2}{2k^2}\right| \,dk
&=
\int \left|  \frac{\widehat{W}_1(k)^2}{2(\tau(k) + \rho a(\CA+\rmu/\rho))} -  \frac{\widehat{W}_1(k)^2}{2k^2}\right| \,dk.
\end{align}
\newcommand{\TTT}{T}
With $T:=\sqrt{2 \constkinetickernel}$ we split this integral into two parts, $|k| \geq \TTT \ell^{-1}$ and the complement.

For $|k| \geq \TTT \ell^{-1}$, we have
\begin{align}\label{eq:LowerBoundByksquare}
\tau(k) = k^2 - \constkinetickernel \ell^{-2},\qquad 
\tau(k) \geq \frac{1}{2} k^2.
\end{align}
Therefore,
\begin{align}
\int_{\{|k|\geq \TTT \ell^{-1}\}} &\left|  \frac{\widehat{W}_1(k)^2}{2(\tau(k) + \rho a(\CA+\rmu/\rho))} -  \frac{\widehat{W}_1(k)^2}{2k^2}\right| \,dk \nonumber \\
&\leq
\widehat{W}_1(0)^2
\int_{\{|k|\geq \TTT \ell^{-1}\}} 
\frac{\left| k^2 - \tau(k) - \rho a(\CA+\rmu/\rho)\right|}{ k^4}\,dk \nonumber \\
&\leq \widehat{W}_1(0)^2 \int_{\{|k|\geq \TTT \ell^{-1}\}} 
\frac{\constkinetickernel \ell^{-2}+ (\CA\rho+\rmu) a }{ k^4}\,dk  \nonumber\\
&= C a \sqrt{\rmu a^3},
\end{align}
where the constant in particular depends on $M_0$ and $K$.

For $|k| \leq \TTT \ell^{-1}$ (and $\rmu$ sufficiently small) we drop $\tau(k)$ and estimate $|\widehat{W}_1(k)| \leq \widehat{W}_1(0)= 8\pi a$. Therefore,
\begin{align}
\int_{\{|k| \leq \TTT \ell^{-1}\}} \left|  \frac{\widehat{W}_1(k)^2}{2(\tau(k) + \rho a(\CA+\rmu/\rho))} -  \frac{\widehat{W}_1(k)^2}{2k^2}\right| \,dk \leq C a \sqrt{\rmu a^3},
\end{align}
with $C$ depending on $K$.
This establishes the bound in \eqref{eq:BoundsOnIntegralsI}.

The bound in \eqref{eq:BoundsOnIntegralsII} is analogous, and we will give fewer details. Notice that
\begin{align}
\ell^{3} \rho^{-1} \frac{\mathcal{B}(k)^4}{\mathcal{A}^3(k)}
=
\rho^2 \frac{\widehat{W}_1^4(k)}{(\tau(k) + \rho a(\CA+\rmu/\rho))^3}.
\end{align}
Upon making the same splitting as for the first integral, we see that for $|k|\leq \TTT \ell^{-1}$ the integrand can be bounded by $C a \rmu^{-1}$ leading to a bound of the right magnitude.
For $|k|\geq \TTT \ell^{-1}$ we again use \eqref{eq:LowerBoundByksquare} and find that the integrand is bounded by $ C\rmu^2 a^4 |k|^{-6}$. Upon explicitly integrating this function we again find a bound of the right magnitude.
\end{proof}
The following lemma was used in the proof of Lemma~\ref{prop:BogIntegral-2-New}.
\begin{lemma}\label{lem:I2-integral}
We have, assuming that  \eqref{eq:R/ell} is satisfied,
\begin{align}\label{eq:I2-integral}
\left|(2\pi)^{-3}  \int \frac{\widehat{W}_1(k)^2}{2k^2} \,dk -
\int g(x) \omega(x)\right| \leq  C a (R/\ell)^2,
\end{align}
for some constant $C$ (depending only on the localization function $\chi$).
\end{lemma}

\begin{proof}
We have $\widehat{\omega}(k)=\frac{\widehat{g}(k)}{2k^2}$ by \eqref{es:scatteringFourier}. 
We calculate the difference between $(2\pi)^{-3} \int \hat{g}(k) \widehat{\omega}(k)\, dk$  
and
$(2\pi)^{-3} \int \frac{\widehat{W}_1^2(k)}{2k^2}\, dk$
using the Fourier transformation and \eqref{eq:W1-g},
\begin{align}
\left| \int \frac{\widehat{W}_1^2(k)- \widehat{g}^2(k)}{k^2} \, dk\right|
&=C_0\left|  \iint \frac{(W_1-g)(x) (W_1+g)(y)}{|x-y|}\,dx\,dy \right| \nonumber \\
&\leq  C_0 C \frac{R^2}{\ell^2} ( 2 +C \frac{R^2}{\ell^2})  \iint \frac{g(x) g(y)}{|x-y|}\,dx\,dy \nonumber \\
&= 2  C \frac{R^2}{\ell^2} ( 2 +C \frac{R^2}{\ell^2}) \int \widehat{g}(k) \widehat{\omega}(k)\,dk.
\label{eq: HLS-not}
\end{align}
Now the lemma follows from the Parseval identity using that $\omega(x) \leq 1$ and that \eqref{eq:R/ell} implies a bound on $R/\ell$.
\end{proof}
\section{Estimating the energy}\label{sec:TotalEnergy}

\begin{lemma}\label{lem:EnergyFirst}
For a given localization function $\chi$ there exists a universal constant $K_0$ so that the following is true for any $K\in (0,K_0]$. If $R/\ell$ is sufficiently small, so that in particular \eqref{eq:R/ell} is satisfied, and if $\Psi$ is normalized and an eigenstate for $n$ with $\rho \leq 20\rmu$, then we have \begin{align}
\ell^{-3}\langle &\Psi, {\mathcal H}_{B}(\rmu) \Psi \rangle \nonumber \\&\geq - 4\pi a\rmu^2
+ \frac{1}{4} (\rho - \rmu)^2 \widehat{g}(0) 
-C  \rmu^2 a  \Big [\left( \rmu a^3\right)^{1/2} +\frac{R^2}{\ell^2}\Big ] .
\end{align}
Here the constant $C$ is allowed to depend on $K$.
\end{lemma}
\begin{proof}
We estimate the term $A_0$ in \eqref{eq:A0} as
\begin{align}
\abs{B}^{-1}\langle \Psi, A_0 \Psi \rangle\geq &\,\frac{1}{2}\rho^2\big(\widehat{g}(0) + \widehat{g\omega}(0)\big)-\Big (\rmu\rho+\frac{1}{4}(\rho-\rmu)^2\Big )\widehat{g}(0)\nonumber\\
&\,-C(\rho\rho_++\abs{B}^{-2}+(\rho+\rmu)\abs{B}^{-1})a.
\end{align}
	Since by assumption $\rho\leq 20 \rho_\mu$, we may apply Lemma~\ref{prop:BogIntegral-2-New} with the choice $M_0=20$. We combine this with the estimates in \eqref{eq:SmallsimpleQs}, \eqref{eq:Def_HBog}, \eqref{eq:EstH1-Bog} and \eqref{eq:Est-H1Bog_integral} and get (using again the bound on $\rho/\rmu$)
	\begin{align}\label{eq:NearConclusion}
	|B|^{-1} \langle \Psi, {\mathcal H}_{B}(\rmu) \Psi \rangle 
	\geq&\,
	- \frac{\widehat{g}(0)}{2} \rmu^2 + \frac{\widehat{g}(0)}{4} (\rho -\rmu)^2
	+( b \ell^{-2} - C'\rmu a) \rho_{+}  \nonumber \\
	&\, -C  \rmu^2 a  \Big [\left( \rmu a^3\right)^{1/2} +\frac{R^2}{\ell^2}\Big ].
	\end{align}
	Here the constant $C$ only depends on $K$ and the localization function $\chi$ while the constant $C'$ is independent of $K$.
	Lemma~\ref{lem:EnergyFirst} now follows using the definition of $\ell$, if we define $K_0$ as the largest choice of $K$ for which the coefficient to $\rho_{+}$ becomes positive.
\end{proof}

Using Lemma~\ref{lem:EnergyFirst} we can now finally give a good estimate on the number $n$ of particles in the box.

\begin{lemma}[Upper bound on $n$]\label{lem:UpperBoundOnn}Let $K\in (0,K_0]$ with $K_0$ as in Lemma~\ref{lem:EnergyFirst}.
If $R/\ell$ is sufficiently small, so that in particular \eqref{eq:R/ell} is satisfied, and $\Psi$ is a normalized eigenvector for $n$ satisfying
\begin{align}
\langle \Psi, {\mathcal H}_{B}(\rmu) \Psi \rangle < 0,
\end{align}
then we have $\rho/\rmu \leq 20$.
\end{lemma}

\begin{proof}
Using Lemma~\ref{lem:EnergyFirst} we only have to consider the case $n \geq 20 \rmu \ell^3$.
We can now split the particles into a number $m$ of groups, each having particle number in the interval $[5  \rmu \ell^3, 20  \rmu \ell^3]$.
Omitting the positive interaction between particles in different groups gives the lower bound
\begin{align}
\langle \Psi, {\mathcal H}_{B}(\rmu) \Psi \rangle
\geq
m {\mathcal G}, 
\end{align}
where
\begin{align*}
{\mathcal G} = \inf\Big\{ \langle \Psi', {\mathcal H}_{B}(\rmu) \Psi' \rangle\,|\,
\Psi' \text{ has $n'$ particles in the box $B$, with } n' \in [5 \rmu \ell^3, 20   \rmu \ell^3] \Big\}.
\end{align*}
We now argue that ${\mathcal G} \geq 0$ if $R/\ell$ is chosen sufficiently small. To see this we insert $\frac{n'}{\ell^3}\geq 5\rmu$ into \eqref{eq:NearConclusion} and note that the resulting coefficient in front of the $\rmu^2$-term becomes positive if $R/\ell$ is sufficiently small. The other terms are either positive or higher powers in $\rmu a^3$. Since $R^2/\ell^2\leq K^{-2}\rmu aR^2 $ this finishes the proof.
\end{proof}

\begin{proof}[Proof of Theorem~\ref{thm:LHY-Box}]\label{proof of thm:LHY-Box}
To obtain Theorem~\ref{thm:LHY-Box} we apply Lemma~\ref{lem:EnergyFirst}, where the value of $K_0$ indirectly is determined, and Lemma~\ref{lem:UpperBoundOnn} to each $N$-particle subspace of ${\mathcal F}_{\rm s}(L^2(\Lambda))$.
\end{proof}
\begin{remark}\label{rm: additional estimates}
Notice that Theorem~\ref{thm:LHY-Box} also could have been formulated for fixed $\chi$ and  $K\in (0,K_0)$. Then,
if $R/\ell$ is sufficiently small, so that in particular \eqref{eq:R/ell} is satisfied, the final arguments in the proof of Theorem~\ref{thm:LHY-Box} imply the bounds
\begin{align}\label{e bound rem 8.3}
0\leq (\frac{n}{|B|}-\rmu)^2a+\rmu\rho_{+}a+ \langle \Psi,{\mathcal Q}_{4}^{\rm ren} \Psi \rangle\abs{B}^{-1}\leq C\rmu^2a\Big ( \sqrt{\rmu a^3}+\frac{R^2}{\ell^2}\Big )
\end{align}
for any normalized $n$-eigenvector, $\Psi$, satisfying 
\begin{align}\label{e condition rem 8.3}
|B|^{-1} \langle \Psi, {\mathcal H}_{B}(\rmu) \Psi \rangle \leq -4\pi a \rmu^2 
+C\rmu^2a\Big ( \sqrt{\rmu a^3}+\frac{R^2}{\ell^2}\Big ),
\end{align}
where the constants in \eqref{e bound rem 8.3} and \eqref{e condition rem 8.3} are allowed to depend on $\chi$ and $K$.
\end{remark}
\appendix
\section{Bogoliubov method}
In this section we recall a simple consequence of the Bogoliubov method.
In \cite{BS} we use the following version (and allow $\cB=-\cA$ if $\kappa=0$)---see also \cite[Theorem~6.3]{LS}.
\begin{theorem}[Simple case of Bogoliubov's method]\label{thm:bogolubov}\hfill\\
	For arbitrary $\cA,\cB\in \R$ satisfying $\cA>0$, $-\cA<\cB\leq\cA$ and $\kappa\in\C$ we have 
	the operator inequality 
	\begin{eqnarray*}
		\cA(b^*_+\bn_++b^*_{-}\bn_{-})+\cB(b^*_+b^*_{-}+\bn_+\bn_{-})+
		\kappa(b^*_++\bn_{-})+\overline{\kappa}(\bn_++b^*_{-})\\ 
		\geq-\frac{1}{2}(\cA-\sqrt{\cA^2-\cB^2})
		([\bn_{+},b^*_{+}]+[\bn_{-},b^*_{-}])-\frac{2|\kappa|^2}{\cA+\cB},
	\end{eqnarray*}
	where $b_\pm$ are operators on a Hilbert space satisfying $[b_+,b_-]=0$. 
\end{theorem}
\bibliographystyle{siam}
%\bibliography{biblio}
\bibliography{LHY-O}

%****************************************
%\bibliographystyle{plain}
%\bibliography{biblio}
\end{document}